\documentclass[12pt,peerreview,draftclsnofoot,letterpaper,oneside,onecolumn]{IEEEtran}

\usepackage{graphicx}
\usepackage{amsfonts}
\usepackage{amsmath}
\usepackage{amssymb}
\usepackage{mathrsfs} 
\usepackage{color}
\usepackage{bm}
\usepackage{latexsym} 
\newtheorem{thm}{Theorem}
\newtheorem{cor}{Corollary}
\newtheorem{lem}{Lemma}
\newtheorem{proof}{proof}

\newtheorem{defn}{Definition}
\newtheorem{rem}{Remark}

\ifCLASSINFOpdf \else \fi
\begin{document}

\newcommand\qbin[3]{\left[\begin{matrix}#1\\#2\end{matrix}\right]_{#3}}

\setcounter{page}{1}

\title{Some New Constructions of Coded Caching Schemes with Reduced
Subpacketization}

\author{Wentu Song, ~ Kui Cai, ~ and  ~ Long Shi
\thanks{The authors are with Singapore University of
Technology and Design, Singapore
       (e-mails: \{wentu$\_$song, cai$\_$kui, shi$\_$long\}@sutd.edu.sg).}
}\maketitle

\begin{abstract}

We study the problem of constructing centralized coded caching
schemes with low subpacketization level based on the
placement delivery array (PDA) design framework. PDA design 
is an efficient way to construct centralized coded caching schemes
and most existing schemes, including the famous Maddah-Ali-Niesen
scheme, can be described using PDA. In this paper, we first prove
that constructing a PDA is equivalent to constructing three binary
matrices that satisfy certain conditions. From this perspective,
we then propose some new constructions of coded caching schemes
using PDA design based on projective geometries over finite
fields, combinatorial configurations, and $t$-designs,
respectively. Our constructions achieve low subpacketization level
(e.g., linear subpacketization) with reasonable rate loss and
include several known results as special cases. Finally, we give
an approach to construct new coded caching scheme from existing
schemes based on direct product of PDAs. Our results enrich the
coded caching schemes of low subpacketization level.
\end{abstract}


\section{Introduction}

Caching is a popular technique to reduce peak traffic rate for
large scalar content delivery over networks. The basic idea is to
allocate fractions of the contents to user memories during the
off-peak traffic times, which can be used to help to serve user
requests during the peak hours and hence reduce network traffic.
Traditional duplication-based caching schemes focused on
exploiting the popularity or statistics of the user demands
$($e.g., see \cite{Meyerson01}$)$. Coded caching was first
proposed in \cite{Maddah-Ali14}, and was shown to be able to
significantly reduce network traffic.

The caching system considered in \cite{Maddah-Ali14} consists of
one server connected by $K$ users through a shared, error-free
link. The server has a database of $N$ files of equal size, and at
certain times each user may demand a specific file from the
server. It is assumed that each user has a cache that allows it to
store $M/N$ fraction of all the files in the server. A
\emph{centralized coded caching scheme} works in two separated
phases: the \emph{placement phase} and the \emph{delivery phase}.
In the placement phase, the server, without any knowledge of the
user demands, allocates certain uncoded packets of the data files
into the cache of users, while in the delivery phase, the server,
upon receiving the specific demands of all users, broadcasts coded
packets through the shared link to all users so that each user can
extract his desired file from the received packets and the cache
content. The main task is to design caching schemes (placement
scheme and corresponding delivery scheme) that minimize the rate
$R$, which is defined as the maximal transmission amount in the
delivery phase among all possible combinations of user demands.
The coded caching scheme proposed in \cite{Maddah-Ali14} attains a
rate of
\begin{align*}
R^*=\frac{K\left(1-\frac{M}{N}\right)}
{1+K\frac{M}{N}},\end{align*} where $1-\frac{M}{N}$ is called the
local caching gain and $1+K\frac{M}{N}$ is called the global
caching gain. Moreover, this rate $R^*$ was proved to be optimal
for uncoded placement \cite{Wan16,Yu17}.

Since the original work of \cite{Maddah-Ali14}, coded caching has
attracted significant attention and many works have been done from
various aspects of the problem. It was shown in \cite{Tian18} that
caching with coded placement can achieve improvement in
memory-rate tradeoff in some regime. Caching where files have
different popularity scores was studied in \cite{Niesen17} and
caching where files have differing sizes was studied in
\cite{Zhang15}. The case of users with heterogenous cache sizes
was investigated in \cite{Wang15} and the case of each user
demanding multiple files was investigated in \cite{Ji15}. The
scenario of decentralized caching was considered in
\cite{Maddah-Ali14-2}, where in the placement phase the cache
content of each user is randomly chosen from the file packets. A
more general case, the hierarchical network with two layers of
caches, was considered in \cite{Karamchandani14}.

\subsection{Related Work}
Although optimal in rate, the caching scheme in
\cite{Maddah-Ali14} has its limitation in practical
implementations: By this caching scheme, each file is divided into
$F=\binom{K}{KM/N}$ packets $($The number $F$ is also referred to
as the file size or subpacketization in some literature.$)$, which
grows exponentially with $K$ \cite{Shanmugam16}. For practical
application, it is important to construct coded caching scheme
with smaller packet number.

So far, several coded caching schemes with reduced file size have
been constructed, all with the sacrifice of increasing the rate.
In \cite{Shanmugam17}, a class of coded caching schemes with
linear file size $($i.e., $F=K)$ were constructed from
Ruzsa-Szem$\acute{\text{e}}$redi graphs. A very interesting
framework for constructing centralized coded caching scheme, named
\emph{placement delivery array} design (or PDA design for
simplicity), was introduced in \cite{Yan17}, and based on the PDA
design framework, some new classes of coded caching schemes were
constructed in \cite{Yan17} and \cite{Cheng19}. Following the PDA
design framework, two classes of constant rate caching schemes
with sub-exponential subpacketization were obtained in
\cite{Chong18} from $(6,3)$-free $3$-partite hypergraphs, and a
more general class of coded caching schemes were constructed in
\cite{Yan18} from strong edge colored bipartite graphs. Coded
caching schemes based on resolvable combinatorial designs from
certain linear block codes were studied in \cite{Tang18}, and
coded caching schemes based on projective geometries over finite
fields were reported in \cite{Krishnan18,Suthan19}. Coded caching
schemes from some other block designs, including balanced
incomplete block designs (BIBDs), $t$-designs and transversal
designs (TDs), are obtained in \cite{Agrawal19}. Summaries of
known centralized coded caching schemes can be found in
\cite{Cheng19,Chong18} and \cite{Krishnan18}.

\subsection{Our Contributions}
In this paper, we focus on constructing centralized coded caching
schemes with reduced subpacketization level. 
Our contributions are as follows.
\begin{itemize}
 \item[1)] We prove that constructing a PDA is equivalent to
 constructing three binary matrices that satisfy some certain
 conditions. This new perspective of PDA design provides a
 convenient way to construct feasible PDAs and the corresponding
 coded caching schemes.
 \item[2)] We propose some new constructions of coded caching schemes
 using PDA design based on projective geometries over finite fields,
 combinatorial configurations and $t$-designs, respectively.
 Our constructions have low subpacketization level (e.g., linear
 subpacketization) at the cost of reasonable rate loss and include several
 known results as special cases.
 \item[3)] We propose a method for constructing
 new coded caching scheme from existing schemes based on a direct
 product technique of PDAs. This direct product method allows us to
 construct more caching schemes from the known results.
\end{itemize}

\subsection{Organization}
The rest of this paper are organized as follows. In Section
\uppercase\expandafter{\romannumeral 2} we introduce the
centralized coded caching problem and other related concepts
including the placement delivery array (PDA) design and some basic
concepts of combinatorial designs. In Section
\uppercase\expandafter{\romannumeral 3} we present a new
perspective of PDA design. Construction of coded caching schemes
based on projective geometries over finite fields is given in
Section \uppercase\expandafter{\romannumeral 4} and constructions
of caching schemes based on configurations and $t$-designs are
given in Section \uppercase\expandafter{\romannumeral 5}. The
direct product of PDAs is given in Section
\uppercase\expandafter{\romannumeral 6}. Finally, the paper is
concluded in Section \uppercase\expandafter{\romannumeral 7}.

\section{Preliminaries}
We use $\mathbb F$ to denote a finite field. If $\mathbb F$ is a
$q$-ary field, we also denote $\mathbb F$ as $\mathbb F_q$. For
any positive integer $n$, denote $[n]:=\{1,2,\cdots,n\}$. For any
set $X$, $|X|$ is the size (cardinality) of $X$. If $|X|=n$, we
call $X$ an $n$-set; if $Y\subseteq X$ and $|Y|=m$, where $0\leq
m\leq |X|$, we call $Y$ an $m$-subset of $ X$. Let $\binom{X}{m}$
denote the set of all $m$-subsets of $X$. For any sets $ X$ and $
Y$, $ X\backslash Y$ is a set consisting of all elements of $ X$
that are not contained in $ Y$. 

An array is denoted by a bold uppercase letter (say \textbf{P},
\textbf{C}, etc), and a vector is denoted by a bold uppercase
letter or a bold lowercase letter (say \textbf{W}, \textbf{Z},
\textbf{d}, etc). If $\textbf{P}=[p_{i,j}]_{m\times n}$ is an $m$
by $n$ array, denote $\textbf{P}(i,j)=p_{i,j}$. We use
$\textbf{C}_{X, Y}$ to denote a matrix whose rows are indexed by
$X$ and columns are indexed by $Y$, where $X$ and $Y$ are two sets
(not necessarily distinct). For each $x\!\in\! X$ and $y\!\in\!
Y$, the element of $\textbf{C}_{X, Y}$ in the $x$-th row and the
$y$-th column is denoted by $\textbf{C}_{X, Y}(x,y)$. We stipulate
that matrices with different subscripts are different, even if the
corresponding subscripts denote the identical set. For example,
$\textbf{C}_{X, Y}$ and $\textbf{C}_{X, Z}$ denote two different
matrices even if $Y$ and $Z$ are identical set.

The following lemma will be used in our discussions.
\begin{lem}\cite[Corollary 5.2]{Bondy}\label{match-B-G}
If $G$ is a $k$-regular bipartite graph with $k>0$, then $G$ has a
perfect matching.
\end{lem}

\subsection{Centralized Coded Caching Problem}


We consider the caching system where one server is connected by
$K$ users through a shared, error-free link. The server has a
library of $N$ files, denoted by
$\textbf{W}_1,\cdots,\textbf{W}_N$, such that each file
$\textbf{W}_i\in\mathbb F^F$, where $\mathbb F$ is some fixed
finite field and $F$ is a positive integer called the
subpacketization. Moreover, we assume that each user $k$ has a
local cache memory of size $M$ with $0\leq M\leq N$, that is, each
user can store a vector $\textbf{Z}_k\in\mathbb F^{MF}$ in its
local cache memory. We refer to such system as a $(K,M,N)$
\emph{caching system}.\footnote{The subpacketization $F$, also
called the file size or the packet number, is not a fixed
parameter of the caching system. In fact, it is determined by the
specific caching scheme.} In this work, it is sufficient to assume
that $\mathbb F=\mathbb F_2$.

The caching system operates in two phases: the placement phase and
the delivery phase. In the placement phase, for each user $k$, the
vector $\textbf{Z}_k$ is computed and allocated into its local
cache memory. In the delivery phase, each user $k$ demands a file
$\textbf{W}_{d_k}$, where $d_k\in[N]$, and the server, informed of
the demands of all users, computes a vector
$\textbf{X}_{\textbf{d}}\in\mathbb F^{\lfloor RF\rfloor}$ and
broadcasts it to all users through the shared link so that each
user $k$ can decode its demanded file $\textbf{W}_{d_k}$ from
$\textbf{Z}_k$ and $\textbf{X}_{\textbf{d}}$, where
$\textbf{d}=(d_1,\cdots,d_K)\in[N]^K$ is called a demand vector
and $R$ is called the rate. As such, an $F$-division \emph{caching
scheme} with rate $R$ is specified by three sets of functions:
\begin{itemize}
 \item[(i)]
 a set of caching functions
 $$\left\{\phi_k:\mathbb F^{NF}\rightarrow
 \mathbb F^{MF}\right\}_{k\in[K]},$$
 \item[(ii)]
 a set of encoding functions
 $$\left\{\varphi_{\textbf{d}}: \mathbb F^{NF}\rightarrow
 \mathbb F^{\lfloor RF\rfloor}\right\}_{\textbf{d}\in[N]^K},$$
 \item[(iii)] a set of decoding functions
 $$\left\{\mu_{k,\textbf{d}}: \mathbb F^{MF}
 \times\mathbb F^{\lfloor RF\rfloor}\rightarrow \mathbb
 F^{NF}\right\}_{k\in[K],\textbf{d}\in[N]^K},$$
\end{itemize}
such that for all $k\in[K]$ and
$\textbf{d}=(d_1,\cdots,d_K)\in[N]^K$,
$$\textbf{W}_{d_k}=\mu_{k,\textbf{d}}
(\textbf{Z}_k,\textbf{X}_{\textbf{d}}),$$ where
$\textbf{Z}_k=\phi_k(\textbf{W}_1,\cdots,\textbf{W}_N)$ and
$\textbf{X}_{\textbf{d}}=\varphi_{\textbf{d}}
(\textbf{W}_1,\cdots,\textbf{W}_N)$. A caching scheme is said to
have uncoded placement if $\textbf{Z}_k$ consists of an exact copy
of some subpackets of $\textbf{W}_1,\cdots,\textbf{W}_N$.
Otherwise, it is said to have coded placement.

\subsection{Placement Delivery Array}

Placement delivery array (PDA) was introduced in \cite{Yan17} and
was shown to be an efficient framework for constructing coded
caching schemes. In this subsection, we briefly review the basic
idea of the PDA framework.

\begin{defn}[PDA]\label{defn-pda}
Let $K,F,Q,S$ be positive integers such that $Q<F$, and
$\textbf{P}=[p_{j,k}]_{F\times K}$ be an $F\times K$ array
consisting of $1,2,\cdots,S$ and an additional symbol $*$.
$\textbf{P}$ is called a $(K,F,Q,S)$ placement delivery array
(PDA) if the following conditions are satisfied:
\begin{itemize}
 \item[C1.] The symbol $*$ appears exactly $Q$ times in each column;
 \item[C2.] Each integer $s\in[S]$ occurs at least once in the array;
 \item[C3.] For any two distinct pairs
 $(j_1,k_1),(j_2,k_2)\in[F]\times[K]$,
 if $p_{j_1,k_1}=p_{j_2,k_2}=s\in[S]$, then $j_1\neq j_2$,
 $k_1\neq k_2$
 and $p_{j_1,k_2}=p_{j_2,k_1}=*$.
\end{itemize}
\end{defn}


\begin{lem}\cite[Theorem 1]{Yan17}\label{lem-pda-ccp}
For any given $(K,F,Q,S)$ PDA $\textbf{P}=[p_{j,k}]_{F\times K}$,
there exists an $F$-division caching scheme for any $(K, M, N)$
caching system with $\frac{M}{N}=\frac{Q}{F}$ and rate
$R=\frac{S}{F}$. Precisely, each user is able to decode its
requested file correctly for any demand $\textbf{d}$.
\end{lem}

Illustrative examples of PDA and details for constructing caching
schemes from PDAs can be found in \cite{Yan17}.

\subsection{Introduction to Combinatorial Designs}

We review some basic concepts of combinatorial designs and several
families of designs that are related to our constructions of
caching schemes. Properties and constructions of such families of
designs can be found in \cite{Stinson,Colbourn}.

\begin{defn}[Design]\label{def-dsgn}
A design is a pair $(V,\mathcal B)$, where $V$ is a set of
elements called points and $\mathcal B$ is a collection (i.e.,
multiset) of nonempty subsets of $V$ called blocks.
\end{defn}

A design $(V,\mathcal B)$ is said to be \emph{simple} if it does
not contain repeated blocks, i.e., $\mathcal B$ is a set (of
blocks) rather than a multiset.

\begin{defn}[Configuration]\label{def-Cfg}
A configuration $(v_{r}, b_k)$ is a design $(V,\mathcal B)$, where
$V$ is a $v$-set of points and $\mathcal B$ is a $b$-set of
blocks, such that the following conditions are satisfied:
\begin{itemize}
 \item [(i)] Each block is a $k$-subset of $V$;
 \item [(ii)] Each point is contained in exactly $r$ blocks;
 \item [(iii)] Every $2$-subset of $V$ is contained in at most
 one block.
\end{itemize}
It is not hard to verify that condition (iii) is equivalent to the
following condition (iii$'$)
\begin{itemize}
 \item [(iii$'$)] Every pair of distinct blocks have at most
 one point in common.
\end{itemize}
\end{defn}

If $(V,\mathcal B)$ is a configuration $(v_{r}, b_k)$, we have
$bk=vr$, or equivalently, $$b=\frac{vr}{k}.$$ Moreover, since
every $2$-subset of $V$ is contained in at most one block, we have
$$\binom{v}{2}\geq b\binom{k}{2}.$$ That is, $$\frac{v(v-1)}{2}\geq
\frac{vr}{k}\frac{k(k-1)}{2},$$ from which we can obtain $$k\leq
\frac{v-1}{r}+1,$$ where equality holds if and only if
$(V,\mathcal B)$ is a $(v,k,1)$-BIBD, which is defined as follows.

\begin{defn}[$(v,k,1)$-BIBD]\label{def-bibd}
Let $v$ and $k$ be positive integers such that $v > k \geq 2$. A
$(v,k,1)$-balanced incomplete block design (abbreviated to
$(v,k,1)$-BIBD) is a design $(V, \mathcal B)$ such that the
following properties are satisfied:
\begin{itemize}
 \item [(i)] $|V|=v$;
 \item [(ii)] Each block is a $k$-subset of $V$;
 \item [(iii)] Every $2$-subset of $V$ is contained in exactly
 one block.
\end{itemize}
\end{defn}

Another important family of designs related to our construction is
$t$-design, which is defined as follows.
\begin{defn}[$t$-design]\label{def-t-dsn}
Let $v,k,\lambda$ and $t$ be positive integers such that $v>k\geq
t$. A $t$-$(v,k,\lambda)$-design is a design $(V,\mathcal B)$
satisfying the following conditions:
\begin{itemize}
 \item [(i)] $|V|=v$;
 \item [(ii)] Each block is a $k$-subset of $V$;
 \item [(iii)] Every $t$-subset of $V$ is contained in exactly
 $\lambda$ blocks.
\end{itemize}
\end{defn}

A $t$-$(v,k,1)$-design is also known as a Steiner system and is
usually denoted by $\text{S}(t,k,v)$.

If $(V,\mathcal B)$ is a $t$-$(v,k,\lambda)$-design and
$Y\subseteq V$ such that $|Y|=s\leq t$, then there are exactly
\begin{align}\label{lmd-s-t-des}
\lambda_s=\frac{\lambda\binom{v-s}{t-s}}{\binom {k-s}
{t-s}}\end{align} blocks in $\mathcal B$ that contain all points
in $Y$. In particular, the number of blocks in a
$t$-$(v,k,\lambda)$-design is
\begin{align*}
b:=\lambda_0=\frac{\lambda\binom{v}{t}}{\binom {k}
{t}}.\end{align*}


\section{A new Perspective of Placement Delivery Array Design}

We now show that designing a feasible placement delivery array is
equivalent to constructing three binary matrices satisfying some
certain conditions.
\begin{thm}\label{thm-PDA-SNC} There exists a $(K,F,Q,S)$ PDA
if and only if there exist three binary matrices $\textbf{C}_{ X,
Y}$, $\textbf{C}_{ X, Z}$ and $\textbf{C}_{ Y, Z}$, where $ X$ is
an $F$-set, $ Y$ is an $S$-set, and $ Z$ is a $K$-set, such that
the following conditions E1$-$E5 are satisfied:
\begin{itemize}
 \item[E1\!.] For every $z\in Z$,
 $$\left|\{x\in  X:
 \textbf{C}_{ X, Z}(x,z)=1\}\right|=|X|-Q.$$
 \item[E2.\!] For every $y\in  Y$,
 $$\left\{z\in\! Z:\textbf{C}_{Y,
 Z}(y,z)=1\right\}\neq\emptyset.$$
 \item[E3.\!] For every $(x,y)\in X\times Y$ such that
 $\textbf{C}_{X,Y}(x,y)=1$, there is exactly one
 $z\in Z$ such that
 $$\textbf{C}_{X,Z}(x,z)=\textbf{C}_{Y,Z}(y,z)=1.$$
 \item[E4.\!] For every $(x,z)\in X\times
 Z$ such that $\textbf{C}_{X,Z}(x,z)=1$, there is exactly one
 $y\in Y$ such that
 $$\textbf{C}_{X, Y}(x,y)=\textbf{C}_{Y, Z}(y,z)=1.$$
 \item[E5.\!] For every $(y,z)\in Y\times Z$ such that
 $\textbf{C}_{Y, Z}(y,z)=1$, there is exactly one
 $x\in X$ such that $$\textbf{C}_{X, Y}(x,y)
 =\textbf{C}_{X, Z}(x,z)=1.$$
\end{itemize}
\end{thm}
\begin{proof}
The proof is given in Appendix A.
\end{proof}

The following theorem provides a set of conditions that are
sufficient for the existence of a feasible PDA and easier to
construct than the conditions in Theorem \ref{thm-PDA-SNC}.

\begin{thm}\label{thm-PDA-SFC} Let $\textbf{C}_{X,
Y}$, $\textbf{C}_{X, Z}$, $\textbf{C}_{Y, Z}$ be three binary
matrices, where $ X$ is an $F$-set, $Y$ is an $S$-set, and $Z$ is
a $K$-set. For every $z\in  Z$, denote
$$U_{z}^{(1)}:=\{x\in  X: \textbf{C}_{ X,
Z}(x,z)=1\}$$ and $$U_{z}^{(2)}:=\{y\in Y: \textbf{C}_{Y,
Z}(y,z)=1\}.$$ Then there exists a $(K,F,Q,S)$ PDA if the
conditions E1-E3 and E6 are satisfied, where E1-E3 are defined as
in Theorem \ref{thm-PDA-SNC}, and E6 is defined as follows.
\begin{itemize}
 \item[E6.\!] For every $z\in  Z$
 and every $(x,y)\in U_{z}^{(1)}\!\times
 U_{z}^{(2)}$,
 $$\left|N_{z}(x)\right|
 =\left|N_{z}(y)\right|\geq 1,$$ where
 $$N_{z}(x):=\{y'\in U_{z}^{(2)}\!:
 \textbf{C}_{X, Y}(x,y')=1\} $$ and
 $$N_{z}(y):=\{x'\in U_{z}^{(1)}:
 \textbf{C}_{X, Y}(x',y)=1\};$$
\end{itemize}
\end{thm}
\begin{proof}
We will construct a binary matrix $\textbf{C}_{X, Y}'$ such that
the matrices $\textbf{C}_{X, Y}'$, $\textbf{C}_{X, Z}$,
$\textbf{C}_{Y, Z}$ satisfy conditions E1-E5.

Consider the bipartite graph $\mathcal G_z$ with bipartition
$(U_{z}^{(1)}, U_{z}^{(2)})$ and edge set
$$E(\mathcal G_z)=\{(x,y)\in U_{z}^{(1)}\times
U_{z}^{(2)}:\textbf{C}_{X, Y}(x,y)=1\}.$$ Clearly, condition E6
implies that $\mathcal G_z$ is $k$-regular, where
$k=\left|N_{z}(x)\right|>0$ for any given $x\in U_{z}^{(1)}$. By
Lemma \ref{match-B-G}, $\mathcal G_z$ has a perfect matching
$\mathcal M_z$. Now, let $\textbf{C}_{ X, Y}'$ be a binary matrix
such that for each $(x,y)\in X\times Y$: $\textbf{C}'_{ X,
Y}(x,y)=1$ if $xy\in\mathcal M_z$ for some $z\in Z$, and
$\textbf{C}'_{ X, Y}(x,y)=0$ otherwise.

By the construction, it is easy to see that $\textbf{C}_{X, Y}'$,
$\textbf{C}_{X, Z}$, $\textbf{C}_{Y, Z}$ satisfy conditions E1-E5.
Hence, by Theorem \ref{thm-PDA-SNC}, there exists a $(K,F,Q,S)$
PDA.
\end{proof}

The following corollary can be viewed as a generalization of
\cite[Theorem 7]{Cheng17}.
\begin{cor}\label{cor-PDA-SFC}
Suppose $\textbf{C}_{X, Y}$, $\textbf{C}_{X, Z}$, $\textbf{C}_{Y,
Z}$ are binary matrices satisfying conditions E3 $($as in Theorem
\ref{thm-PDA-SNC}$)$, E6 $($as in Theorem \ref{thm-PDA-SFC}$)$, as
well as the following conditions E1$'$, E2$'$ and E7:
\begin{itemize}
 \item[E1$'$.\!\!] ~For every $z\in Z$,
 $$\left|\{x\in X\!: \textbf{C}_{X, Z}(x,z)=1\}\right|=D_Z>0.$$
 \item[E2$'$.\!\!] ~For every $y\in Y$,
 $$\left|\{z\in Z\!: \textbf{C}_{Y, Z}(y,z)=1\}\right|=D_Y>0.$$
 \item[E7.] ~For every $x\in X$,
 $$\left|\{z\in Z\!: \textbf{C}_{X, Z}(x,z)=1\}\right|=D_X>0.$$
\end{itemize}
Then there exists a $(K,F,Q,S)$ PDA for each of the following
three sets of parameters:
\begin{itemize}
 \item[1)] $K=|X|$, $F=|Y|$, $Q=|Y|-D_X$, and $S=|Z|$;
 \item[2)] $K=|X|$, $F=|Z|$, $Q=|Z|-D_X$, and $S=|Y|$;
 \item[3)] $K=|Z|$, $F=|X|$, $Q=|X|-D_Z$, and $S=|Y|$;
\end{itemize}
\end{cor}
\begin{proof}
First, as in the proof of Theorem \ref{thm-PDA-SFC}, we can
construct binary matrices $\textbf{C}_{X, Y}'$, $\textbf{C}_{X,
Z}$, $\textbf{C}_{Y, Z}$ satisfying conditions E3$-$E5. Clearly,
$\textbf{C}_{X, Y}'$, $\textbf{C}_{X, Z}$, $\textbf{C}_{Y, Z}$
still satisfy conditions E1$'$, E2$'$ and E7, and hence we can
easily verify that they satisfy conditions E1-E5 of Theorem
\ref{thm-PDA-SNC}, where $Q=|X|-D_Z$. Therefore we can obtain a
$(K,F,Q,S)$ PDA with $K=|Z|$, $F=|X|$, $Q=|X|-D_Z$, and $S=|Y|$,
which has parameter set 3).

Now, relabelling $X$ by $Z$, $Y$ by $X$, and $Z$ by $Y$, we can
check that conditions E1-E5 of Theorem \ref{thm-PDA-SNC} are still
satisfied by $\textbf{C}_{X, Y}'$, $\textbf{C}_{X, Z}$ and
$\textbf{C}_{Y, Z}$. Therefore, we obtain a $(K,F,Q,S)$ PDA with
$K=|X|$, $F=|Y|$, $Q=|Y|-D_X$, and $S=|Z|$, which has parameter
set 1).

Similarly, relabelling $X$ by $Z$, and $Z$ by $X$, we can obtain a
$(K,F,Q,S)$ PDA with $K=|X|$, $F=|Y|$, $Q=|Y|-D_X$, and $S=|Z|$,
which has parameter set 2).
\end{proof}

In subsequent sections, we will construct binary matrices
$\textbf{C}_{X, Y}$, $\textbf{C}_{X, Z}$ and $\textbf{C}_{Y, Z}$
satisfying the desired conditions using projective geometries over
finite fields, configurations and $t$-designs, respectively. Then,
from such matrices, we can construct the corresponding PDAs and
coded caching schemes.

\section{Coded Caching Based on Projective Geometries over
Finite Fields} As a simple application of Corollary
\ref{cor-PDA-SFC}, we give a construction of PDAs based on
projective geometries over finite fields. Our method gives a
different proof of \cite[Theorem 3]{Krishnan18} and two more
families of PDAs.

We first need to introduce the notation of Gaussian binomial
coefficient and the concept of projective space over finite
fields. Let $q$ be a prime power and $\mathbb F_q$ be the $q$-ary
field. For any non-negative integers $\ell$ and $m$ such that
$0\leq m\leq \ell$, the Gaussian binomial coefficient (or
$q$-binomial coefficient), denoted by $\qbin{\ell}{m}{q}$, is
defined as
$$\qbin{\ell}{m}{q}\!=\!\prod_{i=0}^{m-1}\!
\frac{q^{\ell-i}\!-\!1}{q^{m-i}\!-\!1}=
\frac{(q^{\ell}\!-\!1)(q^{\ell-1}\!-\!1)\cdots(q^{\ell-m+1}\!-\!1)}
{(q^{m}\!-\!1)(q^{m-1}\!-\!1)\cdots(q\!-\!1)}.$$ Clearly,
$\qbin{\ell}{\ell}{q}=\qbin{\ell}{0}{q}=1$. Some other useful
properties of Gaussian binomial coefficient, which can be found in
\cite{Hirschfeld}, are listed in the following lemma.
\begin{lem}\label{lem-G-coef}
Suppose $0\leq s\leq m, t\leq \ell$. For any fixed
$\ell$-dimensional vector space $V$ over $\mathbb F_q$, the
following hold.
\begin{itemize}
 \item[1)] The number of $m$-dimensional subspaces of $V$ is
 $\qbin{\ell}{m}{q}$.
 \item[2)] The number of $m$-dimensional subspaces of $V$
 that contain a fixed $s$-dimensional subspace of $V$ is
 $\qbin{\ell-s}{m-s}{q}.$
 \item[3)] The number of $m$-dimensional subspaces of $V$
 intersecting a fixed $t$-dimensional subspace of $V$ in a fixed
 $s$-dimensional subspace is $q^{(m-s)(t-s)}\qbin{\ell-t}{m-s}{q}$.
\end{itemize}
\end{lem}

Now, we can give a construction of PDA based on projective
geometries over finite fields. Consider positive integers $k,m$
and $t$ such that $m+t\leq k$. Let $X$ be the set of all
$t$-dimensional subspaces of $\mathbb F_q^k$, $Y$ be the set of
all $m$-dimensional subspaces of $\mathbb F_q^k$, and $Z$ be the
set of all $(m+t)$-dimensional subspaces of $\mathbb F_q^k$. By 1)
of Lemma \ref{lem-G-coef}, we have $|X|=\qbin{k}{t}{q}$,
$|Y|=\qbin{k}{m}{q}$ and $|Z|=\qbin{k}{m+t}{q}$. We construct
three binary matrices $\textbf{C}_{X,Y}$, $\textbf{C}_{X,Z}$ and
$\textbf{C}_{Y,Z}$ as follows. For every $x\in X$ and $y\in Y$:
\begin{equation}\label{eq1-CS-PG}
\textbf{C}_{X,Y}(x,y)=\left\{\begin{aligned} &1&
\text{if}~ x\cap y=\{\textbf{0}\}, \\
&0& \text{otherwise,} ~ ~ ~ ~ ~ ~\\
\end{aligned} \right.
\end{equation}
where $\textbf{0}$ is the zero vector of $\mathbb F_q^k$.
Moreover, for $z\in Z$:
\begin{equation}\label{eq2-CS-PG}
\textbf{C}_{X,Z}(x,z)=\left\{\begin{aligned} &1&
\text{if}~x\subseteq z,~ \\
&0& \text{otherwise,} \\
\end{aligned} \right.
\end{equation}
and
\begin{equation}\label{eq3-CS-PG}
\textbf{C}_{Y,Z}(y,z)=\left\{\begin{aligned} &1&
\text{if}~y\subseteq z,~ \\
&0& \text{otherwise.} \\
\end{aligned} \right.
\end{equation}
We can prove that $\textbf{C}_{X,Y}$, $\textbf{C}_{X,Z}$ and
$\textbf{C}_{Y,Z}$ satisfy conditions E1$'$, E2$'$, E3, E6 and E7
of Corollary \ref{cor-PDA-SFC} by the following discussions:
\begin{itemize}
 \item Given an $m$-dimensional subspace $x$ of $\mathbb F_q^k$ and
 a $t$-dimensional subspace $y$ of $\mathbb F_q^k$ such that $x\cap
 y=\{\textbf{0}\}$, there is exactly one
 $(m+t)$-dimensional subspace $z$ of $\mathbb
 F_q^k$, i.e., $z=x\oplus y$ is the direct sum
 of $x$ and $y$, such that $x\subseteq z$ and $y\subseteq z$.
 So by \eqref{eq1-CS-PG}$-$\eqref{eq3-CS-PG}, Condition E3 is satisfied.
 \item Given an $(m+t)$-dimensional subspace $z$ of $\mathbb F_q^k$
 and a $t$-dimensional subspace $x$ of $\mathbb F_q^k$ such that
 $x\subseteq z$, by 3) of Lemma \ref{lem-G-coef}, there are $q^{mt}$
 $m$-dimensional subspaces $y$
 of $\mathbb F_q^k$ such that $y\subseteq z$ and $x\cap y=\{\textbf{0}\}$.
 Similarly, given an $(m+t)$-dimensional subspace $z$ of
 $\mathbb F_q^k$ and an
 $m$-dimensional subspace $y$ of $\mathbb F_q^k$ such that
 $y\subseteq z$, there are $q^{mt}$ $t$-dimensional subspaces $x$
 of $\mathbb F_q^k$ such that $x\subseteq z$ and
 $x\cap y=\{\textbf{0}\}$. By \eqref{eq1-CS-PG}$-$\eqref{eq3-CS-PG},
 for every $z\in Z$ every $x\in U_z^{(1)}$ and $y\in U_z^{(2)}$,
 $|N_z(x)|=|N_z(y)|=q^{mt}$, so condition E6 is satisfied.
 \item For every $t$-dimensional subspace $x$ of $\mathbb
 F_q^k$, by 2) of Lemma \ref{lem-G-coef},
 there are $\qbin{k-t}{m}{q}$~ $(m+t)$-dimensional subspaces $z$
 of $\mathbb F_q^k$ such that $x\subseteq z$. So by \eqref{eq2-CS-PG},
 condition E7 of Corollary
 \ref{cor-PDA-SFC} is satisfied and we have $D_X=\qbin{k-t}{m}{q}$.
 Similarly, we can verify that condition E2$'$ is satisfied and
 $D_Y=\qbin{k-m}{t}{q}$.
 \item For every $(m+t)$-dimensional subspace $z$ of $\mathbb
 F_q^k$, by 1) of Lemma \ref{lem-G-coef},
 there are $\qbin{m+t}{t}{q}$~ $t$-dimensional subspaces $x$
 of $\mathbb F_q^k$ such that $x\subseteq z$. So by \eqref{eq2-CS-PG},
 condition E1$'$ of Corollary \ref{cor-PDA-SFC} is satisfied and
 $D_Z=\qbin{m+t}{t}{q}$.
\end{itemize}
By Corollary \ref{cor-PDA-SFC}, there exists a $(K,F,Q,S)$ PDA,
where $(K,F,Q,S)$ can be any of the following three cases:

Case 1. $(K,F,Q,S)=(|X|,|Y|,|Y|-D_X,|Z|)$. Then by the above
discussions, we have \vspace{-2pt}\begin{align*}
(K,F,Q,S)=\left(\qbin{k}{t}{q}, \qbin{k}{m}{q},
\qbin{k}{m}{q}-\qbin{k-t}{m}{q},
\qbin{k}{m+t}{q}\right);\end{align*}

Case 2. $(K,F,Q,S)=(|X|,|Z|,|Z|-D_X,|Y|)$. In this case, we have
\vspace{-2pt}
\begin{align*}
(K,F,Q,S)=\left(\qbin{k}{t}{q}, \qbin{k}{m+t}{q},
\qbin{k}{m+t}{q}-\qbin{k-t}{m}{q},
\qbin{k}{m}{q}\right);\end{align*}

Case 3. $(K,F,Q,S)=(|Z|,|X|,|X|-D_Z,|Y|)$. In this case, we have
\vspace{-2pt}
\begin{align*}
(K,F,Q,S)=\left(\qbin{k}{m+t}{q}, \qbin{k}{t}{q},
\qbin{k}{t}{q}-\qbin{m+t}{t}{q},
\qbin{k}{m}{q}\right).\end{align*}

By Lemma \ref{lem-pda-ccp}, from the PDAs in the above three
cases, we can obtain correspondingly three families of caching
schemes as stated by the following theorem.
\begin{thm}\label{CS-PG}
Let $k,m$ and $t$ be positive integers satisfying $m+t\leq k$.
Then there exists a caching scheme for each of the following three
sets of parameters.
\begin{itemize}
 \item[1)] $K=\qbin{k}{t}{q}$,
 $\frac{M}{N}=1-\frac{\qbin{k-t}{m}{q}}{\qbin{k}{m}{q}}$,
 $R=\frac{\qbin{k}{m+t}{q}}{\qbin{k}{m}{q}}$, and
 $F=\qbin{k}{m}{q}$;
 \vspace{5pt}\item[2)] $K=\qbin{k}{t}{q}$,
 $\frac{M}{N}=1-\frac{\qbin{k-t}{m}{q}}{\qbin{k}{m+t}{q}}$,
 $R=\frac{\qbin{k}{m}{q}}{\qbin{k}{m+t}{q}}$, and
 $F=\qbin{k}{m+t}{q}$;
 \vspace{5pt}\item[3)] $K=\qbin{k}{m+t}{q}$, $\frac{M}{N}=1-
 \frac{\qbin{m+t}{t}{q}}{\qbin{k}{t}{q}}$,
 $R=\frac{\qbin{k}{m}{q}}{\qbin{k}{t}{q}}$ and
 $F=\qbin{k}{t}{q}$.
\end{itemize}
\end{thm}

\vspace{5pt}It is easy to verify that
$1-\frac{\qbin{k-t}{m}{q}}{\qbin{k}{m+t}{q}}=
1-\frac{\qbin{m+t}{t}{q}}{\qbin{k}{t}{q}}$ and
$\frac{\qbin{k}{m}{q}}{\qbin{k}{m+t}{q}}=
\frac{\qbin{m+t}{t}{q}}{\qbin{k-m}{t}{q}}$. So we can see that the
construction in \cite[Theorem 3]{Krishnan18} coincide with our
construction in Theorem \ref{CS-PG} with parameter set 2).

\section{Coded Caching Based on Combinatorial Designs}
In this section, we present some caching schemes constructed from
configurations and $t$-designs.

\subsection{Caching Schemes from Configurations}
We first give a construction based on combinatorial
configurations.
\begin{thm}\label{CS-cfg}
If there exists a configuration $(v_{r}, b_k)$, then there exists
a caching scheme for each of the following three sets of
parameters.
\begin{itemize}
 \item[1)] $K=F=v$, $\frac{M}{N}=1-\frac{r}{v}$, and
 $R=\frac{r}{k}$;
 \item[2)] $K=v$, $\frac{M}{N}=1-\frac{k}{v}$,
 $R=\frac{k}{r}$, and $F=\frac{vr}{k}$;
 \item[3)] $K=\frac{vr}{k}$, $\frac{M}{N}=1-\frac{k}{v}$,
 $R=1$, and $F=v$.
\end{itemize}
\end{thm}
\begin{proof}
Let $(V,\mathcal B)$ be a configuration $(v_{r}, b_k)$. Let
$X=Y=V$ and $Z=\mathcal B$. Then $|X|=|Y|=v$ and
$|Z|=b=\frac{vr}{k}$. We construct three binary matrices
$\textbf{C}_{X,Y}$, $\textbf{C}_{X,Z}$ and $\textbf{C}_{Y,Z}$ as
follows. For every $x\in X$ and $y\in Y$:
\begin{equation*}
\textbf{C}_{X,Y}(x,y)\!=\!\left\{\begin{aligned} &1& \!\text{if}~
x\!\neq\!
y~\text{and}~\{x, y\}\!\subseteq\! B~\text{for some}
~B\!\in\!\mathcal B,\\
&0& \text{otherwise.} ~ ~ ~ ~ ~ ~ ~ ~ ~ ~ ~ ~ ~ ~ ~ ~ ~ ~ ~ ~ ~ ~
~ ~ ~ ~ ~ ~ ~ ~ ~ ~ ~ ~ ~ \\
\end{aligned} \right.
\end{equation*}
Moreover, for every $z\in Z$:
\begin{equation*}
\textbf{C}_{X,Z}(x,z)=\left\{\begin{aligned} &1&
\text{if}~x\in z,~ \\
&0& \text{otherwise,} \\
\end{aligned} \right.
\end{equation*}
and
\begin{equation*}
\textbf{C}_{Y,Z}(y,z)=\left\{\begin{aligned} &1&
\text{if}~y\in z,~ \\
&0& \text{otherwise.} \\
\end{aligned} \right.
\end{equation*}

Note that $(V,\mathcal B)$ is a configuration $(v_{r}, b_k)$. We
can prove that $\textbf{C}_{X,Y}$, $\textbf{C}_{X,Z}$ and
$\textbf{C}_{Y,Z}$ satisfy conditions E1$'$, E2$'$, E3, E6 and E7
of Corollary \ref{cor-PDA-SFC} by the following discussions:
\begin{itemize}
 \item Every block $z\in \mathcal B$ contains $k$ points of $V$, so by the
 construction of $\textbf{C}_{X,Z}$, $D_Z=k$ and Condition E1$'$ is satisfied.
 \item Every point $y$ is contained by $r$ blocks, so by the construction
 of $\textbf{C}_{Y,Z}$, $D_Y=r$ and Condition E2$'$ is satisfied.
 Similarly, $D_X=r$ and Condition E7 is satisfied.
 \item Since in a configuration $(v_{r}, b_k)$, every $2$-subset of $V$
 is contained in at most one block, so for every pair of distinct points
 $\{x,y\}$, if there is a $B\in \mathcal B$ such that $\{x, y\}\subseteq B$,
 then $B$ is the unique block that contains both $x$ and $y$. By the
 construction of $\textbf{C}_{X,Y}$, $\textbf{C}_{X,Z}$ and
 $\textbf{C}_{Y,Z}$, we can see that condition E3 of Corollary
 \ref{cor-PDA-SFC} is satisfied.
 \item By the construction of  $\textbf{C}_{X,Y}$, $\textbf{C}_{X,Z}$
 and $\textbf{C}_{Y,Z}$, for every block $z\in\mathcal B$ and
 every $x\in U_z^{(1)}$, we have $x\in z$ and $N_z(x)=z\backslash\{x\}$;
 Similarly, for every $y\in U_z^{(2)}$, we have $y\in z$ and
 $N_z(y)=z\backslash\{y\}$. So $|N_z(x)|=|N_z(y)|=k-1$ and
 condition E6 is satisfied.
\end{itemize}
By Corollary \ref{cor-PDA-SFC}, there exists a $(K,F,Q,S)$ PDA,
where $(K,F,Q,S)$ can be any of the following three cases:

Case 1. $(K,F,Q,S)=\left(v, v, v-r, \frac{vr}{k}\right)$. Then by
Lemma \ref{lem-pda-ccp}, we can obtain a caching scheme with
$K=F=v$, $\frac{M}{N}=1-\frac{r}{v}$, and the rate
$R=\frac{r}{k}$, which has parameter set 1).

Case 2. $(K,F,Q,S)=\left(v, \frac{vr}{k}, \frac{vr}{k}-r,
v\right)$. Then by Lemma \ref{lem-pda-ccp}, we can obtain a
caching scheme with $K=v$, $\frac{M}{N}=1-\frac{k}{v}$,
$R=\frac{k}{r}$, and $F=\frac{vr}{k}$, which has parameter set 2).

Case 3. $(K,F,Q,S)=\left(\frac{vr}{k}, v, v-k, v\right)$. Then by
Lemma \ref{lem-pda-ccp}, we can obtain a caching scheme with
$K=\frac{vr}{k}$, $\frac{M}{N}=1-\frac{k}{v}$, $R=1$, and $F=v$,
which has parameter set 3).
\end{proof}

Consider the construction in Theorem \ref{CS-cfg} with parameter
set 1). We have $K=F=v$ and $r=K\left(1-\frac{M}{N}\right)$. Since
in a configuration,
$$k\leq \frac{v-1}{r}+1=\frac{K-1}{K\left(1-\frac{M}{N}\right)}+1,$$
then we can obtain \begin{align*}R=\frac{r}{k}\geq
\frac{K\left(1-\frac{M}{N}\right)}
{\frac{K-1}{K\left(1-\frac{M}{N}\right)}+1}
=\frac{K^2\left(1-\frac{M}{N}\right)^2}
{K-1+K\left(1-\frac{M}{N}\right)},\end{align*} where equality
holds if and only if the corresponding configuration is a
$(v,k,1)$-BIBD. Thus, if there exists a $(v,k,1)$-BIBD, then by
Theorem \ref{CS-cfg}, we can obtain a caching scheme with linear
subpacketization $(K=F)$ and, for fixed $\frac{M}{N}$ and large
$K$, we have
\begin{align*}R=\frac{K^2\left(1-\frac{M}{N}\right)^2}
{K-1+K\left(1-\frac{M}{N}\right)}\approx
\frac{K^2\left(1-\frac{M}{N}\right)^2}
{K+K\left(1-\frac{M}{N}\right)}=K\left(1-\frac{M}{N}\right)
\frac{\left(1-\frac{M}{N}\right)}{2-\frac{M}{N}},\end{align*}
which grows linearly with $K$. For the general case, if $r$ and
$v$ are given, then $K,F$ and $\frac{M}{N}$ are determined and to
reduce the rate $R$, we are interested in finding the largest $k$
such that there exists a configuration $(v_{r}, b_k)$, where
$b=\frac{vr}{k}$.

\begin{rem}
We can compare our construction with the rate-optimal caching
scheme in \cite{Maddah-Ali14}. Note that the optimal rate
$R^*=K\left(1-\frac{M}{N}\right) \frac{1}{1+K\frac{M}{N}}$ is
achieved using an exponentially growing subpacketization
$F=\binom{K}{K\frac{M}{N}}$. In comparison, our construction uses
linear subpacketization $($i.e., $F=K)$ and the rate $R$ satisfies
$$\frac{R}{R^*}=\left(1+K\frac{M}{N}\right)
\frac{\left(1-\frac{M}{N}\right)}{2-\frac{M}{N}}.$$ That is, the
rate of our construction differs from the optimal rate by a factor
of $\left(1+K\frac{M}{N}\right)
\frac{\left(1-\frac{M}{N}\right)}{2-\frac{M}{N}}$. Note that by
\cite[Theorem 12]{Chong18}, given $R=\frac{S}{F}$ and
$\frac{M}{N}=\frac{Q}{F}$ independent of $K$, a $(K, F, Q, S)$ PDA
where $F$ grows linearly with $K$ does not exist. So the loss in
$R=K\left(1-\frac{M}{N}\right)
\frac{\left(1-\frac{M}{N}\right)}{2-\frac{M}{N}}$ is reasonable.
\end{rem}


\begin{rem}
We can also compare our constructions in Theorem \ref{CS-cfg} with
some results in \cite{Agrawal19}.
\begin{itemize}
 \item[1)] Note that a $(v,k,1)$-BIBD is a configuration $(v_{r},
 b_k)$ with $r=\frac{v-1}{k-1}$ and
 $b=\frac{vr}{k}=\frac{v(v-1)}{k(k-1)}$. We can easily check that
 \cite[Theorem 2]{Agrawal19} is just a special case of our construction
 in Theorem \ref{CS-cfg} with parameter set 2).
 \item[2)] Note that a TD$(k,n)$, whose definition can be found in
 \cite[Chapter 6]{Stinson}, is a configuration $(v_{r},
 b_k)$ with $r=n$, $v=kn$ and $b=\frac{vr}{k}=n^2$.
 So we can check that \cite[Theorem 6]{Agrawal19} is just a special
 case of our construction in Theorem \ref{CS-cfg} with parameter set 3).
 Moreover, \cite[Theorem 6]{Agrawal19} is proved only
 for $k\geq n$, while our Theorem \ref{CS-cfg} holds for all
 possible $k$ and $n$.
\end{itemize}
\end{rem}

\subsection{Caching Schemes from $t$-$(v,k,1)$-Designs}

We present two constructions of caching schemes from
$t$-$(v,k,1)$-designs. The first construction comes from
$t$-$(v,k,1)$-designs with $t\leq\frac{k}{2}+1$ (see Theorem
\ref{CS-tdsn-1}) and the second construction comes from
$t$-$(v,k,1)$-designs with $\frac{k}{2}+1<t\leq k$ (see Theorem
\ref{CS-tdsn-2}). Constructing caching schemes using $t$-$(v,k,1)$
was also investigated in \cite{Agrawal19}. Our constructions are
different from the constructions in \cite{Agrawal19}, and hence
the resulted caching schemes have different parameter sets.

\begin{thm}\label{CS-tdsn-1}
Suppose there exists a $t$-$(v,k,1)$-design with
$t\leq\frac{k}{2}+1$. Let $t_0$ be a positive integer such that
$\frac{t}{2}\leq t_0\leq t-1$. Then there exists a caching scheme
for each of the following three sets of parameters.
\begin{itemize}
 \item[1)] $K=F=\binom{v}{t_0}$,
 $\frac{M}{N}=1-\frac{\binom{v-t_0}{t-t_0}}{\binom {k-t_0}
 {t-t_0}\binom{v}{t_0}}$, and
 $R=\frac{\binom{v}{t}}{\binom{k}{t}\binom{v}{t_0}}$;
 \vspace{3pt}\item[2)] $K=\binom{v}{t_0}$,
 $\frac{M}{N}=1-\frac{\binom{k}{t_0}}{\binom{v}{t_0}}$,
 $R=\frac{\binom{v}{t_0}\binom
 {k}{t}}{\binom {v}{t}}$, and
 $F=\frac{\binom{v}{t}}{\binom{k}{t}}$;
 \vspace{3pt}\item[3)] $K=\frac{\binom{v}{t}}{\binom {k}{t}}$,
 $\frac{M}{N}=1-\frac{\binom{k}{t_0}}{\binom{v}{t_0}}$,
 $R=1$, and $F=\binom{v}{t_0}$.
\end{itemize}
\end{thm}
\begin{proof}
Let $(V,\mathcal B)$ be a $t$-$(v,k,1)$-design, where $t\leq
\frac{k}{2}+1$. Let $X=Y=\binom{V}{t_0}$ and $Z=\mathcal B$. Then
$|X|=|Y|=\binom{v}{t_0}$ and $|Z|=b=\frac{\binom{v}{t}}{\binom
{k}{t}}$. We construct three binary matrices $\textbf{C}_{X,Y}$,
$\textbf{C}_{X,Z}$ and $\textbf{C}_{Y,Z}$ as follows. For every
$x\in X$ and $y\in Y$:
\begin{equation*}
\textbf{C}_{X,Y}(x,y)\!=\!\left\{\begin{aligned} &1& \!\text{if}~
x\cap y=\emptyset~\text{and}~(x\cup y)\!\subseteq\! B
~\text{for some}~B\!\in\!\mathcal B, \\
\vspace{5pt}&0& \text{otherwise.} ~ ~ ~ ~ ~ ~ ~ ~ ~ ~ ~ ~ ~ ~ ~ ~
~ ~ ~ ~ ~ ~ ~ ~ ~ ~ ~ ~ ~ ~ ~ ~ ~ ~ ~ ~ ~ ~ ~ ~ ~ ~ \\
\end{aligned} \right.
\end{equation*}
Moreover, for every $z\in Z$:
\begin{equation*}
\textbf{C}_{X,Z}(x,z)=\left\{\begin{aligned} &1&
\text{if}~x\subseteq z,~ \\
&0& \text{otherwise,} \\
\end{aligned} \right.
\end{equation*}
and
\begin{equation*}
\textbf{C}_{Y,Z}(y,z)=\left\{\begin{aligned} &1&
\text{if}~y\subseteq z,~ \\
&0& \text{otherwise.} \\
\end{aligned} \right.
\end{equation*}

Note that $(V,\mathcal B)$ is a $t$-$(v,k,1)$-design. We can prove
that $\textbf{C}_{X,Y}$, $\textbf{C}_{X,Z}$ and $\textbf{C}_{Y,Z}$
satisfy conditions E1$'$, E2$'$, E3, E6 and E7 of Corollary
\ref{cor-PDA-SFC} by the following discussions:
\begin{itemize}
 \item Every block $z\in \mathcal B$ is a $k$-subset of $V$, so
 $z$ contains $\binom{k}{t_0}$~ $t_0$-subsets of $V$. By the
 construction of $\textbf{C}_{X,Z}$, $D_Z=\binom{k}{t_0}$ and
 condition E1$'$ is satisfied.
 \item In a $t$-$(v,k,1)$-design, every $t_0$-subset $y$ of $V$ is
 contained by $\lambda_{t_0}$ blocks. So by the construction
 of $\textbf{C}_{Y,Z}$, $D_Y=\lambda_{t_0}=\frac{\binom{v-t_0}{t-t_0}}
 {\binom {k-t_0}{t-t_0}}$ and condition E2$'$ is satisfied.
 Similarly, $D_X=\lambda_{t_0}=\frac{\binom{v-t_0}{t-t_0}}{\binom {k-t_0}
 {t-t_0}}$ and condition E7 is satisfied.
 \item Since $t\leq\frac{k}{2}+1$ and $\frac{t}{2}\leq t_0\leq t-1$,
 if $x$, $y$ are $t_0$-subsets of $V$ such that $x\cap y=\emptyset$,
 then $t\leq|x\cup y|=2t_0\leq 2t-2\leq k$. Noticing that in a
 $t$-$(v,k,1)$-design, every $t$-subset of $V$
 is contained in exactly one block, so if there is a $B\in \mathcal B$
 such that $x\cup y\subseteq B$,
 then $B$ is the unique block that contains both $x$ and $y$. By the
 construction of $\textbf{C}_{X,Y}$, $\textbf{C}_{X,Z}$ and
 $\textbf{C}_{Y,Z}$, we can see that condition E3 of Corollary
 \ref{cor-PDA-SFC} is satisfied.
 \item By the construction of $\textbf{C}_{X,Y}$, $\textbf{C}_{X,Z}$
 and $\textbf{C}_{Y,Z}$, for every block $z\in\mathcal B$ and
 every $x\in U_z^{(1)}$, we have $x\subseteq z$ and
 $N_z(x)=\binom{z\backslash x}{t_0}$;
 Similarly, for every $y\in U_z^{(2)}$, we have $y\subseteq z$ and
 $N_z(y)=\binom{z\backslash y}{t_0}$. So $|N_z(x)|=|N_z(y)|
 =\binom{k-t_0}{t_0}$ and condition E6 is satisfied.
\end{itemize}
By Corollary \ref{cor-PDA-SFC}, there exists a $(K,F,Q,S)$ PDA,
where $(K,F,Q,S)$ can be any of the following three cases:

Case 1. $(K,F,Q,S)=\left(\binom{v}{t_0}, \binom{v}{t_0},
\binom{v}{t_0}-\lambda_{t_0}, b\right)$. Then by Lemma
\ref{lem-pda-ccp}, we can obtain a caching scheme with
$K=F=\binom{v}{t_0}$,
$\frac{M}{N}=1-\frac{\lambda_{t_0}}{\binom{v}{t_0}}
=1-\frac{\binom{v-t_0}{t-t_0}}{\binom {k-t_0}
{t-t_0}\binom{v}{t_0}}$, and the rate
$R=\frac{b}{\binom{v}{t_0}}=\frac{\binom{v}{t}}{\binom
{k}{t}\binom{v}{t_0}}$, which has parameter set 1).

Case 2. $(K,F,Q,S)=\left(\binom{v}{t_0}, b, b-\lambda_{t_0},
\binom{v}{t_0}\right)$. Then by Lemma \ref{lem-pda-ccp}, we can
obtain a caching scheme with $K=\binom{v}{t_0}$,
$\frac{M}{N}=1-\frac{\lambda_{t_0}}{b}
=1-\frac{\binom{v-t_0}{t-t_0}\binom{k}{t}}{\binom {k-t_0}
{t-t_0}\binom{v}{t}}=1-\frac{\binom{k}{t_0}}{\binom{v}{t_0}}$, the
rate $R=\frac{\binom{v}{t_0}}{b}=\frac{\binom{v}{t_0}\binom
{k}{t}}{\binom {v}{t}}$, and $F=b=\frac{\binom{v}{t}}{\binom
{k}{t}}$, which has parameter set 2).

Case 3. $(K,F,Q,S)=\left(b, \binom{v}{t_0},
\binom{v}{t_0}-\binom{k}{t_0}, \binom{v}{t_0}\right)$. Then by
Lemma \ref{lem-pda-ccp}, we can obtain a caching scheme with
$K=b=\frac{\binom{v}{t}}{\binom {k}{t}}$,
$\frac{M}{N}=1-\frac{\binom{k}{t_0}}{\binom{v}{t_0}}$, the rate
$R=\frac{\binom{v}{t_0}}{\binom{v}{t_0}}=1$, and the file size
$F=\binom{v}{t_0}$, which has parameter set 3).
\end{proof}

Consider the caching scheme with parameter set 1) in Theorem
\ref{CS-tdsn-1}. Noticing that $K=\binom{v}{t_0}$ and $\frac{M}{N}
=1-\frac{\binom{v-t_0}{t-t_0}}{\binom {k-t_0}
{t-t_0}\binom{v}{t_0}}$, then we can obtain
\begin{align}\label{eq1-tdsn-1}
K\left(1-\frac{M}{N}\right)=\frac{\binom{v-t_0}{t-t_0}}{\binom
{k-t_0} {t-t_0}}.\end{align} Moreover, since
$R=\frac{\binom{v}{t}}{\binom {k}{t}\binom{v}{t_0}}$, then we can
obtain
\begin{align}\label{eq2-tdsn-1}
KR=\binom{v}{t_0}R=\frac{\binom{v}{t}}{\binom {k}{t}}.\end{align}
From \eqref{eq1-tdsn-1} and \eqref{eq2-tdsn-1}, we have
$$\frac{1-\frac{M}{N}}{R}=\frac{\binom{v-t_0}{t-t_0}}{\binom
{k-t_0}{t-t_0}}\frac{\binom{k}{t}}{\binom
{v}{t}}=\frac{\binom{k}{t_0}}{\binom
{v}{t_0}}=\frac{\binom{k}{t_0}}{K},$$ and hence the delivery rate
$$R=K\left(1-\frac{M}{N}\right)\frac{1}{\binom{k}{t_0}}.$$
Unfortunately, it looks hard to find an expression of
$\binom{k}{t_0}$ in terms of $K$ and $\frac{M}{N}$.

The following theorem gives another construction of coded caching
scheme using $t$-$(v,k,1)$-design.
\begin{thm}\label{CS-tdsn-2}
Suppose there exists a $t$-$(v,k,1)$-design such that
$\frac{k}{2}+1<t\leq k$. Let $t_1, t_2$ be positive integers such
that $\max\{t_1,t_2\}<t$ and $t_1+t_2=k$. Then there exists a
caching scheme for each of the following three sets of parameters.
\begin{itemize}
 \item[1)] $K=\binom{v}{t_1}$,
 $\frac{M}{N}=1-\frac{\binom{v-t_1}{t-t_1}}{\binom {k-t_1}
 {t-t_1}\binom{v}{t_2}}$,
 $R=\frac{\binom{v}{t}}{\binom
 {k}{t}\binom{v}{t_2}}$, and $F=\binom{v}{t_2}$;
 \vspace{3pt}\item[2)] $K=\binom{v}{t_1}$,
 $\frac{M}{N}=1-\frac{\binom{k}{t_1}}{\binom{v}{t_1}}$,
 $R=\frac{\binom{v}{t_2}\binom
 {k}{t}}{\binom {v}{t}}$, and $F=\frac{\binom{v}{t}}{\binom
 {k}{t}}$;
 \vspace{3pt}\item[3)] $K=\frac{\binom{v}{t}}{\binom {k}{t}}$,
 $\frac{M}{N}=1-\frac{\binom{k}{t_1}}{\binom{v}{t_1}}$,
 $R=\frac{\binom{v}{t_2}}{\binom{v}{t_1}}$, and $F=\binom{v}{t_1}$.
\end{itemize}
\end{thm}
\begin{proof}
Let $(V,\mathcal B)$ be a $t$-$(v,k,1)$-design, where
$\frac{k}{2}+1<t\leq k$. Let $X=\binom{V}{t_1}$,
$Y=\binom{V}{t_2}$ and $Z=\mathcal B$, where $\max\{t_1,t_2\}<t$
and $t_1+t_2=k$. Then $|X|=\binom{v}{t_1}$, $|Y|=\binom{v}{t_2}$,
and $|Z|=b=\frac{\binom{v}{t}}{\binom {k}{t}}$. We construct three
binary matrices $\textbf{C}_{X,Y}$, $\textbf{C}_{X,Z}$ and
$\textbf{C}_{Y,Z}$ as follows. For every $x\in X$ and $y\in Y$:
\begin{equation}\label{eq1-CS-tdsn-2}
\textbf{C}_{X,Y}(x,y)=\left\{\begin{aligned} &1& \!\text{if}~
x\cap y=\emptyset~\text{and}~(x\cup y)\in\!\mathcal B, \\
\vspace{5pt}&0& \text{otherwise.} ~ ~ ~ ~ ~ ~ ~ ~ ~ ~ ~ ~ ~ ~ ~ ~
~ ~ ~ ~ ~ ~ ~ \\
\end{aligned} \right.
\end{equation}
Moreover, for every $z\in Z$:
\begin{equation}\label{eq2-CS-tdsn-2}
\textbf{C}_{X,Z}(x,z)=\left\{\begin{aligned} &1&
\text{if}~x\subseteq z,~ \\
&0& \text{otherwise,} \\
\end{aligned} \right.
\end{equation}
and
\begin{equation}\label{eq3-CS-tdsn-2}
\textbf{C}_{Y,Z}(y,z)=\left\{\begin{aligned} &1&
\text{if}~y\subseteq z,~ \\
&0& \text{otherwise.} \\
\end{aligned} \right.
\end{equation}
We can verify that conditions E1$'$, E2$'$, E3, E6 and E7 of
Corollary \ref{cor-PDA-SFC} are satisfied as follows.
\begin{itemize}
 \item Note that for a $t$-$(v,k,1)$-design $V,\mathcal B)$,
 every block $z\in\mathcal B$ is a $k$-subset of $V$, and so there
 are $\binom{k}{t_1}$ $t_1$-subsets $x$ of $V$ that are contained by $z$.
 By the construction in
 \eqref{eq2-CS-tdsn-2}, $D_Z=\binom{k}{t_1}$ and E1$'$ is satisfied.
 \item For a $t$-$(v,k,1)$-design $(V,\mathcal B)$, by \eqref{lmd-s-t-des},
 every $t_2$-subset $y$ of $V$ is contained by $\lambda_{t_2}$ blocks. So
 by the construction in \eqref{eq3-CS-tdsn-2}, we have
 $D_Y=\lambda_{t_2}=\frac{\binom{v-t_2}{t-t_2}}{\binom
 {k-t_2}{t-t_2}}$ and E2$'$ is satisfied.
 Similarly, for every $x\in X=\binom{V}{t_1}$, we have
 $D_X=\lambda_{t_1}=\frac{\binom{v-t_1}{t-t_1}}{\binom
 {k-t_1}{t-t_1}}$ is the number of blocks in $\mathcal B$ that contain
 $x$, and so E7 is satisfied.
 \item Note that $t_1+t_2=k$. Then in a $t$-$(v,k,1)$-design
 $(V,\mathcal B)$,
 for every $t_1$-subset $x$ of $V$ and every $t_2$-subset $y$ of
 $V$ such that $x\cap y=\emptyset$, $z=x\cup y$ is the unique
 block in $\mathcal B$ that contains both $x$ and $y$.
 By \eqref{eq1-CS-tdsn-2}$-$\eqref{eq3-CS-tdsn-2}, for every
 $x\in X=\binom{V}{t_1}$ and $y\in Y=\binom{V}{t_2}$
 such that $\textbf{C}_{X,Y}(x,y)=1$, $z=x\cup y$ is the
 unique $z\in Z=\mathcal B$ such that $\textbf{C}_{X,Z}(x,z)=
 \textbf{C}_{Y,Z}(y,z)=1$. So condition E3 is satisfied.
 \item Also note that $t_1+t_2=k$. Then for every fixed block
 $z\in \mathcal B$ and every $t_1$-subset
 $x$ of $z$, $y'=z\backslash x$ is the unique $t_2$-subset of $z$ such
 that $x\cap y'=\emptyset$. By \eqref{eq1-CS-tdsn-2}$-$\eqref{eq3-CS-tdsn-2},
 for every fixed $z\in Z=\mathcal B$ and every $x\in U^{(1)}_z$,
 $N_z(x)=\{y'=z\backslash x\}$. Similarly, for every $y\in U^{(2)}_z$,
 $N_z(y)=\{x'=z\backslash y\}$. So $|N_z(x)|=|N_z(y)|=1$, and hence
 condition E6 is satisfied.
\end{itemize}
By Corollary \ref{cor-PDA-SFC}, there exists a $(K,F,Q,S)$ PDA,
where $(K,F,Q,S)$ can be any of the following three cases:

Case 1. $(K,F,Q,S)=\left(\binom{v}{t_1}, \binom{v}{t_2},
\binom{v}{t_2}-\lambda_{t_1}, b\right)$. Then by Lemma
\ref{lem-pda-ccp}, we can obtain a caching scheme with
$K=\binom{v}{t_1}$,
$\frac{M}{N}=1-\frac{\lambda_{t_1}}{\binom{v}{t_2}}
=1-\frac{\binom{v-t_1}{t-t_1}}{\binom {k-t_1}
{t-t_1}\binom{v}{t_2}}$, the rate
$R=\frac{b}{\binom{v}{t_2}}=\frac{\binom{v}{t}}{\binom
{k}{t}\binom{v}{t_2}}$, and the file size $F=\binom{v}{t_2}$,
which has parameter set 1).

Case 2. $(K,F,Q,S)=\left(\binom{v}{t_1}, b, b-\lambda_{t_1},
\binom{v}{t_2}\right)$. Then by Lemma \ref{lem-pda-ccp}, we can
obtain a caching scheme with $K=\binom{v}{t_1}$,
$\frac{M}{N}=1-\frac{\lambda_{t_1}}{b}
=1-\frac{\binom{v-t_1}{t-t_1}\binom{k}{t}}{\binom {k-t_1}
{t-t_1}\binom{v}{t}}=1-\frac{\binom{k}{t_1}}{\binom{v}{t_1}}$, the
rate $R=\frac{\binom{v}{t_2}}{b}=\frac{\binom{v}{t_2}\binom
{k}{t}}{\binom {v}{t}}$, and $F=b=\frac{\binom{v}{t}}{\binom
{k}{t}}$, which has parameter set 2).

Case 3. $(K,F,Q,S)=\left(b, \binom{v}{t_1},
\binom{v}{t_1}-\binom{k}{t_1}, \binom{v}{t_2}\right))$. Then by
Lemma \ref{lem-pda-ccp}, we can obtain a caching scheme with
$K=b=\frac{\binom{v}{t}}{\binom {k}{t}}$,
$\frac{M}{N}=1-\frac{\binom{k}{t_1}}{\binom{v}{t_1}}$, the rate
$R=\frac{\binom{v}{t_2}}{\binom{v}{t_1}}$, and the file size
$F=\binom{v}{t_1}$, which has parameter set 3).
\end{proof}

\begin{rem}
Note that the $t$-$(v,k,1)$-design with $t=k$ always exists. In
fact, it is just the design $(V,\mathcal B)$ such that $\mathcal
B$ is the set of all $k$-subsets of $V$. It is easy to see that
the construction in Theorem \ref{CS-tdsn-2} with parameter set 1)
for the special case of $t=k$ coincides with the construction of
\cite[Theorem 14]{Chong18}. Hence, our construction in Theorem
\ref{CS-tdsn-2} includes the construction \cite[Theorem
14]{Chong18} as a special case.
\end{rem}

\subsection{Caching Schemes from $t$-$(v,k,\lambda)$-Designs}

We give a construction method of coded caching scheme based on PDA
design using $t$-$(v,k,\lambda)$-designs with $\lambda\geq 1$,
which is unlike Theorem \ref{CS-tdsn-1} and Theorem
\ref{CS-tdsn-2}, where $\lambda=1$.

\begin{thm}\label{CS-tdsn-3}
Suppose there exists a $t$-$(v,k,\lambda)$ design with
$\lambda\geq 1$. Consider positive integers $t_0, t_1$ and $t_2$
such that $t_0=t_1+t_2\leq t$. Then there exists a caching scheme
for each of the following three sets of parameters.
\begin{itemize}
 \item[1)] $K=\frac{\lambda\binom{v}{t_1}\binom{v-t_1}{t-t_1}}
 {\binom{k-t_1}{t-t_1}}$,
 $\frac{M}{N}=1-\frac{\binom{k-t_1}{t_2}\binom{k-t_2}{t-t_2}}
{\lambda\binom{v}{t_2}\binom{v-t_2}{t-t_2}}$,
 $R=\frac{\binom{k}{t_0}}{\binom{k}{t_2}}$, and
 $F=\frac{\lambda\binom{v}{t_2}\binom{v-t_2}{t-t_2}}
 {\binom{k-t_2}{t-t_2}}$;
 \vspace{3pt}\item[2)] $K=\frac{\lambda\binom{v}{t_1}\binom{v-t_1}{t-t_1}}
 {\binom{k-t_1}{t-t_1}}$,
 $\frac{M}{N}=1-\frac{\binom{k-t_1}{t_2}\binom{k-t_0}{t-t_0}}
{\lambda\binom{v}{t_0}\binom{v-t_0}{t-t_0}}$,
 $R=\frac{\binom{k}{t_2}}{\binom{k}{t_0}}$, and
 $F=\frac{\lambda\binom{v}{t_0}\binom{v-t_0}{t-t_0}}
 {\binom{k-t_0}{t-t_0}}$;
 \vspace{3pt}\item[3)] $K=\frac{\lambda\binom{v}{t_0}\binom{v-t_0}{t-t_0}}
 {\binom{k-t_0}{t-t_0}}$,
 $\frac{M}{N}=1-\frac{\binom{t_0}{t_1}\binom{k-t_1}{t-t_1}}
{\lambda\binom{v}{t_1}\binom{v-t_1}{t-t_1}}$,
 $R=\frac{\binom{k}{t_2}}{\binom{k}{t_1}}$, and
 $F=\frac{\lambda\binom{v}{t_1}\binom{v-t_1}{t-t_1}}
 {\binom{k-t_1}{t-t_1}}$.
\end{itemize}
\end{thm}
\begin{proof}
Let $(V,\mathcal B)$ be a $t$-$(v,k,\lambda)$ design, where
$\lambda\geq 1$. For each $\alpha\subseteq V$ of size
$|\alpha|\leq t$, we use $\mathcal B_{\alpha}$ to denote the
collection of blocks containing $\alpha$. Let
$$X=\left\{(\alpha_1,B_1): \alpha_1\in\binom{V}{t_1},
B_1\in\mathcal B_{\alpha_1}\right\},$$
$$Y=\left\{(\alpha_2,B_2): \alpha_2\in\binom{V}{t_2},
B_2\in\mathcal B_{\alpha_2}\right\},$$ and
$$Z=\left\{(\alpha_0,B_0): \alpha_0\in\binom{V}{t_0},
B_0\in\mathcal B_{\alpha_0}\right\}.$$ Then by
\eqref{lmd-s-t-des},
\begin{align*}|X|=\binom{v}{t_1}\lambda_{t_1}
=\binom{v}{t_1}\frac{\lambda\binom{v-t_1}{t-t_1}}{\binom {k-t_1}
{t-t_1}},\end{align*}
\begin{align*}|Y|=\binom{v}{t_2}\lambda_{t_2}
=\binom{v}{t_2}\frac{\lambda\binom{v-t_2}{t-t_2}}{\binom {k-t_2}
{t-t_2}},\end{align*} and
\begin{align*}|Z|=\binom{v}{t_0}\lambda_{t_0}
=\binom{v}{t_0}\frac{\lambda\binom{v-t_0}{t-t_0}}{\binom {k-t_0}
{t-t_0}}.\end{align*}

We can further construct three binary matrices $\textbf{C}_{X,Y}$,
$\textbf{C}_{X,Z}$ and $\textbf{C}_{Y,Z}$ as follows. For every
$x=(\alpha_1,B_1)\in X$ and every $y=(\alpha_2,B_2)\in Y$:
\begin{equation*}
\textbf{C}_{X,Y}(x,y)=\left\{\begin{aligned} &1&
\text{if}~\alpha_1\cap\alpha_2=\emptyset ~\text{and}~ B_1=B_2, \\
&0& \text{otherwise.} ~ ~ ~ ~ ~ ~ ~ ~ ~ ~ ~ ~ ~ ~ ~ ~ ~ ~ ~ ~ ~ ~ ~ \\
\end{aligned} \right.
\end{equation*}
Moreover, for every $z=(\alpha_0,B_0)\in Z$:
\begin{equation*}
\textbf{C}_{X,Z}(x,z)=\left\{\begin{aligned} &1&
\text{if}~\alpha_1\subseteq\alpha_0, ~\text{and}~ B_1=B_0, \\
&0& \text{otherwise,} ~ ~ ~ ~ ~ ~ ~ ~ ~ ~ ~ ~ ~ ~ ~ ~ ~ ~ ~ ~ \\
\end{aligned} \right.
\end{equation*}
and
\begin{equation*}
\textbf{C}_{Y,Z}(y,z)=\left\{\begin{aligned} &1&
\text{if}~\alpha_2\subseteq\alpha_0, ~\text{and}~ B_2=B_0, \\
&0& \text{otherwise.} ~ ~ ~ ~ ~ ~ ~ ~ ~ ~ ~ ~ ~ ~ ~ ~ ~ ~ ~ ~ \\
\end{aligned} \right.
\end{equation*}
The following discussions show that $\textbf{C}_{X,Y}$,
$\textbf{C}_{X,Z}$ and $\textbf{C}_{Y,Z}$ satisfy conditions
E1$'$, E2$'$, E3, E6 and E7 of Corollary \ref{cor-PDA-SFC}:
\begin{itemize}
 \item Note that $t_0=t_1+t_2$.
 Given a $t_1$-subset $\alpha_1$ of $V$, a $t_2$-subset
 $\alpha_2$ of $V$, and a block
 $B\in\mathcal B$ such that $\alpha_1\cap\alpha_2=\emptyset$,
 $\alpha_1\subseteq B$ and $\alpha_2\subseteq B$, there exists
 exactly one $t_0$-subset $\alpha_0$ of $V$, i.e.,
 $\alpha_0=\alpha_1\cup\alpha_2$, such that
 $\alpha_1\subseteq\alpha_0\subseteq B$ and
 $\alpha_2\subseteq\alpha_0\subseteq B$.
 So by the construction of $\textbf{C}_{X,Y}$,
 $\textbf{C}_{X,Z}$ and $\textbf{C}_{Y,Z}$, condition E3 of
 Corollary \ref{cor-PDA-SFC} is satisfied.
 \item Given a $t_1$-subset $\alpha_1$ of $V$, a $t_0$-subset
 $\alpha_0$ of $V$, and a block
 $B\in\mathcal B$ such that $\alpha_1\subseteq\alpha_0\subseteq B$,
 there exists exactly one $t_2$-subset $\alpha_2$ of $V$, i.e.,
 $\alpha_2=\alpha_0\backslash\alpha_1$, such that
 $\alpha_1\cap\alpha_2=\emptyset$ and
 $\alpha_2\subseteq\alpha_0\subseteq B$.
 By the construction of
 $\textbf{C}_{X,Y}$, $\textbf{C}_{X,Z}$ and $\textbf{C}_{Y,Z}$,
 we have
 $N_z(x)=\{(\alpha_0\backslash\alpha_1,B)\}$ for every
 $z=(\alpha_0,B)\in Z$ and every $x=(\alpha_1,B)\in U^{(1)}_z$.
 Similarly, we can obtain
 $N_z(y)=\{(\alpha_0\backslash\alpha_2,B)\}$ for every
 $z=(\alpha_0,B)\in Z$ and every
 $y=(\alpha_2,B)\in U^{(2)}_z$. So $|N_z(x)|=|N_z(y)|=1$, and hence
 condition E6 is satisfied.
 \item For every $\alpha_1\in\binom{V}{t_1}$ and
 $B_1\in\mathcal B_{\alpha_1}$, there is a one-to-one correspondence
 between the $t_2$-subsets of $B_1\backslash\alpha_1$ and the
 $t_0$-subsets of $B_1$ that contain $\alpha_1~($noticing that
 $t_0=t_1+t_2)$.
 So there are $\binom{k-t_1}{t_2}$ $t_0$-subsets $\alpha_0$ of $V$
 such that $\alpha_1\subseteq\alpha_0\subseteq B_1$.
 Hence, by the construction of $\textbf{C}_{X,Z}$, condition E7
 is satisfied and $D_X=\binom{k-t_1}{t_2}.$
 Similarly, by the construction of $\textbf{C}_{Y,Z}$,
 we can verify that condition E2$'$ is satisfied and $D_Y=\binom{k-t_2}{t_1}.$
 \item For every $\alpha_0\in\binom{V}{t_0}$ and
 $B_0\in\mathcal B_{\alpha_0}$,
 there are $\binom{t_0}{t_1}$~ $t_1$-subsets $\alpha_1$ of $V$
 such that $\alpha_1\subseteq\alpha_0\subseteq B_0$.
 So by the construction of $\textbf{C}_{X,Z}$, condition E1$'$
 is satisfied and $D_Z=\binom{t_0}{t_1}.$
\end{itemize}
By Corollary \ref{cor-PDA-SFC}, there exists a $(K,F,Q,S)$ PDA,
where $(K,F,Q,S)$ can be any of the following three cases.

Case 1. $(K,F,Q,S)=\left(|X|,|Y|,|Y|-D_X,|Z|\right)$. Then by
Lemma \ref{lem-pda-ccp}, we can obtain a caching scheme with
\begin{align*}
K=|X|=\binom{v}{t_1}\lambda_{t_1}
=\frac{\lambda\binom{v}{t_1}\binom{v-t_1}{t-t_1}}{\binom{k-t_1}{t-t_1}},
\end{align*}
\begin{align*}
F=|Y|=\binom{v}{t_2}\lambda_{t_2}
=\frac{\lambda\binom{v}{t_2}\binom{v-t_2}{t-t_2}}{\binom{k-t_2}{t-t_2}},
\end{align*}
\begin{align*}
\frac{M}{N}=\frac{Q}{F}
=1-\frac{\binom{k-t_1}{t_2}}{\binom{v}{t_2}\lambda_{t_2}}
=1-\frac{\binom{k-t_1}{t_2}\binom{k-t_2}{t-t_2}}
{\lambda\binom{v}{t_2}\binom{v-t_2}{t-t_2}},
\end{align*}
and
\begin{align*}
R=\frac{S}{F}
=\frac{\binom{v}{t_0}\lambda_{t_0}}{\binom{v}{t_2}\lambda_{t_2}}
=\frac{\binom{v}{t_0}\binom{v-t_0}{t-t_0}\binom{k-t_2}{t-t_2}}
{\binom{v}{t_2}\binom{k-t_0}{t-t_0}\binom{v-t_2}{t-t_2}}=
\frac{\binom{k}{t_0}}{\binom{k}{t_2}},
\end{align*}
which has parameter set 1).

Case 2. $(K,F,Q,S)=\left(|X|,|Z|,|Z|-D_X,|Y|\right)$. Then by
Lemma \ref{lem-pda-ccp}, we can obtain a caching scheme with
\begin{align*}
K=|X|=\binom{v}{t_1}\lambda_{t_1}
=\frac{\lambda\binom{v}{t_1}\binom{v-t_1}{t-t_1}}{\binom{k-t_1}{t-t_1}},
\end{align*}
\begin{align*}
F=|Z|=\binom{v}{t_0}\lambda_{t_0}
=\frac{\lambda\binom{v}{t_0}\binom{v-t_0}{t-t_0}}{\binom{k-t_0}{t-t_0}},
\end{align*}
\begin{align*}
\frac{M}{N}=\frac{Q}{F}
=1-\frac{\binom{k-t_1}{t_2}}{\binom{v}{t_0}\lambda_{t_0}}
=1-\frac{\binom{k-t_1}{t_2}\binom{k-t_0}{t-t_0}}
{\lambda\binom{v}{t_0}\binom{v-t_0}{t-t_0}},
\end{align*}
and
\begin{align*}
R=\frac{S}{F}
=\frac{\binom{v}{t_2}\lambda_{t_2}}{\binom{v}{t_0}\lambda_{t_0}}
=\frac{\binom{v}{t_2}\binom{k-t_0}{t-t_0}\binom{v-t_2}{t-t_2}}
{\binom{v}{t_0}\binom{v-t_0}{t-t_0}\binom{k-t_2}{t-t_2}} =
\frac{\binom{k}{t_2}}{\binom{k}{t_0}},
\end{align*}
which has parameter set 2).

Case 3. $(K,F,Q,S)=\left(|Z|,|X|,|X|-D_Z,|Y|\right)$. Then by
Lemma \ref{lem-pda-ccp}, we can obtain a caching scheme with
\begin{align*}
K=|Z|=\binom{v}{t_0}\lambda_{t_0}
=\frac{\lambda\binom{v}{t_0}\binom{v-t_0}{t-t_0}}{\binom{k-t_0}{t-t_0}},
\end{align*}
\begin{align*}
F=|X|=\binom{v}{t_1}\lambda_{t_1}
=\frac{\lambda\binom{v}{t_1}\binom{v-t_1}{t-t_1}}{\binom{k-t_1}{t-t_1}},
\end{align*}
\begin{align*}
\frac{M}{N}=\frac{Q}{F}
=1-\frac{\binom{t_0}{t_1}}{\binom{v}{t_1}\lambda_{t_1}}
=1-\frac{\binom{t_0}{t_1}\binom{k-t_1}{t-t_1}}
{\lambda\binom{v}{t_1}\binom{v-t_1}{t-t_1}},
\end{align*}
and
\begin{align*}
R=\frac{S}{F}
=\frac{\binom{v}{t_2}\lambda_{t_2}}{\binom{v}{t_1}\lambda_{t_1}}
=\frac{\binom{v}{t_2}\binom{k-t_1}{t-t_1}\binom{v-t_2}{t-t_2}}
{\binom{v}{t_1}\binom{v-t_1}{t-t_1}\binom{k-t_2}{t-t_2}}=
\frac{\binom{k}{t_2}}{\binom{k}{t_1}},
\end{align*}
which has parameter set 3).
\end{proof}

\section{Direct Product of PDAs}
In this section, we give a method to construct new PDA from known
PDAs using the direct product operation. This direct product
method allows us to construct more caching schemes from the
existing results.

\begin{thm}\label{CS-Drt-Prd}
If there exist a $(K_1, F_1, Q_1, S_1)$ PDA and a $(K_2, F_2, Q_2,
S_2)$ PDA, then there exists a $(K, F, Q, S)$ PDA such that
$K=K_1K_2$, $F=F_1F_2$, $Q=F_1Q_2+F_2Q_1-Q_1Q_2$, and $S=S_1S_2$.
\end{thm}
\begin{proof}
By Theorem \ref{thm-PDA-SNC}, from the $(K_1, F_1, Q_1, S_1)$ PDA,
we can obtain three binary matrices $\textbf{C}_{X_1, Y_1}$,
$\textbf{C}_{X_1, Z_1}$ and $\textbf{C}_{Y_1, Z_1}$, where $X_1$
is an $F_1$-set, $Y_1$ is an $S_1$-set, and $Z_1$ is a $K_1$-set
such that conditions E1-E5 are satisfied. Similarly, from the
$(K_2, F_2, Q_2, S_2)$ PDA, we can obtain binary matrices
$\textbf{C}_{X_2, Y_2}$, $\textbf{C}_{X_2, Z_2}$ and
$\textbf{C}_{Y_2, Z_2}$, where $X_2$ is an $F_2$-set, $Y_2$ is an
$S_2$-set, and $Z_2$ is a $K_2$-set such that conditions E1-E5 are
satisfied. Let $X=X_1\times X_2$, $Y=Y_1\times Y_2$ and
$Z=Z_1\times Z_2$. Then $|X|=F_1F_2$, $Y=S_1S_2$, and $Z=K_1K_2$.
Further, construct binary matrices $\textbf{C}_{X, Y}$,
$\textbf{C}_{X, Z}$ and $\textbf{C}_{Y, Z}$ as follows. For every
$x=(x_1,x_2)\in X$ and $y=(y_1,y_2)\in Y$, where $x_1\in X_1$,
$x_2\in X_2$, $y_1\in Y_1$ and $y_2\in Y_2$:
\begin{equation*}
\textbf{C}_{X,Y}(x,y)=\left\{\begin{aligned} &1&
\!\text{if}~\textbf{C}_{X_1,Y_1}(x_1,y_1)\!
=\!\textbf{C}_{X_2,Y_2}(x_2,y_2)\!=\!1, \\
&0& \text{otherwise.} ~ ~ ~ ~ ~ ~ ~ ~ ~ ~ ~ ~ ~ ~ ~ ~
~ ~ ~ ~ ~ ~ ~ ~ ~ ~ ~ ~ ~ ~ ~ ~ ~ ~\\
\end{aligned} \right.
\end{equation*}
Moreover, for every $z=(z_1,z_2)\in Z$, where $z_1\in Z_1$ and
$z_2\in Z_2$:
\begin{equation*}
\textbf{C}_{X,Z}(x,z)=\left\{\begin{aligned} &1&
\!\text{if}~\textbf{C}_{X_1,Z_1}(x_1,z_1)\!
=\!\textbf{C}_{X_2,Z_2}(x_2,z_2)\!=\!1, \\
&0& \text{otherwise.} ~ ~ ~ ~ ~ ~ ~ ~ ~ ~ ~ ~ ~ ~ ~ ~ ~ ~ ~ ~ ~ ~
~ ~ ~ ~ ~ ~ ~ ~ ~ ~ ~ ~\\\end{aligned} \right.
\end{equation*}
and
\begin{equation*}
\textbf{C}_{Y,Z}(y,z)=\left\{\begin{aligned} &1&
\!\text{if}~\textbf{C}_{Y_1,Z_1}(y_1,z_1)\!
=\!\textbf{C}_{Y_2,Z_2}(y_2,z_2)\!=\!1, \\
&0& \text{otherwise.} ~ ~ ~ ~ ~ ~ ~ ~ ~ ~ ~ ~ ~ ~ ~ ~ ~ ~ ~ ~ ~ ~
~ ~ ~ ~ ~ ~ ~ ~ ~ ~ ~\\\end{aligned} \right.
\end{equation*}

Since $\textbf{C}_{X_1, Y_1}$, $\textbf{C}_{X_1, Z_1}$ and
$\textbf{C}_{Y_1, Z_1}$ satisfy conditions E1, for each $z_1\in
Z_1$, we have $$\left|\{x_1\in X_1: \textbf{C}_{X_1,
Z_1}(x_1,z_1)=1\}\right|=|X_1|-Q_1.$$ Similarly, for each $z_2\in
Z_2$, we have $$\left|\{x_2\in X_2: \textbf{C}_{X_2,
Z_2}(x_2,z_2)=1\}\right|=|X_2|-Q_2.$$ Then by the construction,
for each $z=(z_1,z_2)\in Z$, \begin{align*}&\left|\{x\in X:
\textbf{C}_{X, Z}(x,z)=1\}\right|\\&=(|X_1|-Q_1)(|X_2|-Q_2)\\
&=|X_1||X_2|-(|X_1Q_2+|X_2|Q_1-Q_1Q_2)\\
&=|X|-(F_1Q_2+F_2Q_1-Q_1Q_1)\\&=|X|-Q.\end{align*} So the binary
matrices $\textbf{C}_{X, Y}$, $\textbf{C}_{X, Z}$ and
$\textbf{C}_{Y, Z}$ satisfy condition E1. Moreover, it is easy to
check that $\textbf{C}_{X, Y}$, $\textbf{C}_{X, Z}$ and
$\textbf{C}_{Y, Z}$ satisfy conditions E2-E5. By Theorem
\ref{thm-PDA-SNC}, there exists a $(K, F, Q, S)$ PDA, where
$$K=|Z|=K_1K_2,$$ $$F=|X|=F_1F_2,$$ $$Q=F_1Q_2+F_2Q_1-Q_1Q_2,$$ and
$$S=|Y|=S_1S_2,$$ which completes the proof.
\end{proof}

We call the resulted $(K, F, Q, S)$ PDA in Theorem
\ref{CS-Drt-Prd} the \emph{direct product} of the original $(K_1,
F_1, Q_1, S_1)$ PDA and $(K_2, F_2, Q_2, S_2)$ PDA. Clearly, the
direct product construction can be applied to more than two PDAs.

The following Corollary is just a direct consequence of Lemma
\ref{lem-pda-ccp} and Theorem \ref{CS-Drt-Prd}.
\begin{cor}\label{cor-Drt-Prd}
Suppose there exist a $(K_1, F_1, Q_1, S_1)$ PDA and a $(K_2, F_2,
Q_2, S_2)$ PDA. Then there exists a caching scheme for any
$(K,M,N)$ caching system with $K=K_1K_2$,
$\frac{M}{N}=\frac{Q_1}{F_1}+\frac{Q_2}{F_2}
-\frac{Q_1}{F_1}\frac{Q_2}{F_2}$, the delivery rate
$R=\frac{S_1}{F_1}\frac{S_2}{F_2}$, and the file size $F=F_1F_2$.
\end{cor}

\begin{rem}
Another method to construct new PDA from old ones was considered
in \cite{Shanmugam16} and \cite{Cheng19-1}, where multiple copies,
say $m$ copies, of a $(K,F,Q,S)$ PDA is concatenated to form a new
$(mK,F,Q,mS)$ PDA. This concatenation method can reduce the
subpacketization level, but the delivery rate increases to the $m$
times of the original rate.
\end{rem}



\section{Conclusions and Discussions}

We present a new perspective of the PDA design problem based on
three mutually related binary matrices. From this new perspective,
the process of PDA design can be conveniently described. We
proposed some new constructions of PDAs and the corresponding
coded caching schemes with low subpacketization level based on
projective geometries over finite fields, combinatorial
configurations, and $t$-designs, which also generalize some known
results in the literature. We also give a direct product method
for constructing new coded caching scheme from existing schemes.
The results of this paper enrich the constructions of coded
caching schemes of low subpacketization level.

\appendices

\section{Proof of Theorem \ref{thm-PDA-SNC}}

\begin{proof}[Proof of Necessity]
Suppose that $\textbf{P}\!=\![p_{j,k}]_{F\times K}$ is a $(K\!,
F\!, Q, S)$ PDA. Let $X=[F]$, $Y=[S]$, and $Z=[K]$. We will
construct three binary matrices $\textbf{C}_{X, Y}$,
$\textbf{C}_{X, Z}$ and $\textbf{C}_{Y, Z}$ that satisfy
conditions E1$-$E5 of Theorem \ref{thm-PDA-SNC}.

The matrix $\textbf{C}_{X, Z}$ can be constructed as follows: For
every $x\in X$ and $z\in Z$,
\begin{equation}\label{eq1-PDA-C2E}
\textbf{C}_{X, Z}(x,z)=\left\{\begin{aligned}
&1& ~ \text{if}~ p_{x,z}\in[S], \\
&0& ~ \text{if}~ p_{x,z}=*.~~ \\
\end{aligned} \right.
\end{equation}
For every $y\in Y=[S]$, denote
$$\Omega_y=\{(x,z)\!\in\! X\times  Z\!:
p_{x,z}\!=\!y\}\!=\!\{(x_1,z_1), \cdots\!,
(x_{\ell_y},z_{\ell_y})\}.$$ Clearly,
$\Omega_y\cap\Omega_{y'}=\emptyset$ for any distinct $y, y'\in Y$,
and by condition C2, $\Omega_y\neq\emptyset$. Moreover, by
condition C3, $x_1, \cdots\!, x_{\ell_y}$ are distinct, and $z_1,
\cdots, z_{\ell_y}$ are distinct. So
$$A_y\!:=\!\{x_1, \cdots\!, x_{\ell_y}\}\neq\emptyset$$
and $$B_y\!:=\{z_1, \cdots, z_{\ell_y}\}\neq\emptyset.$$ Then the
matrix $\textbf{C}_{X, Y}$ can be constructed as: For every $x\in
X$ and $y\in Y$,
\begin{equation}\label{eq2-PDA-C2E}
\textbf{C}_{X, Y}(x,y)=\left\{\begin{aligned}
&1& ~ \text{if}~ x\in A_y, \\
&0& ~ \text{otherwise}.\\
\end{aligned} \right.
\end{equation}
Further, $\textbf{C}_{Y, Z}$ can be constructed as: For every
$y\in Y$ and $z\in Z$,
\begin{equation}\label{eq3-PDA-C2E}
\textbf{C}_{Y, Z}(y,z)=\left\{\begin{aligned}
&1& ~ \text{if}~ z\in B_y, \\
&0& ~ \text{otherwise}.\\
\end{aligned} \right.
\end{equation}
It remains to prove that $\textbf{C}_{X, Y}$, $\textbf{C}_{X, Z}$
and $\textbf{C}_{Y, Z}$ satisfy conditions E1$-$E5 of Theorem
\ref{thm-PDA-SNC}.

Since $\textbf{P}$ is a $(K\!, F\!, Q, S)$ PDA, by condition C1
and by \eqref{eq1-PDA-C2E}, we have $$|\{x\in X\!: \textbf{C}_{X,
Z}(x,z)=1\}|=F-Q$$ for every $z\in Z$. So $\textbf{C}_{X, Y}$,
$\textbf{C}_{X, Z}$ and $\textbf{C}_{Y, Z}$ satisfy condition E1.

For every $y\in Y$, by \eqref{eq3-PDA-C2E}, we have $$\left\{z\in
Z:\textbf{C}_{Y, Z}(y,z)=1\right\}=B_y\neq\emptyset.$$ So
$\textbf{C}_{X, Y}$, $\textbf{C}_{X, Z}$ and $\textbf{C}_{Y, Z}$
satisfy condition E2.

For any given $(x,y)\in X\times Y$ such that $\textbf{C}_{X,
Y}(x,y)=1$, by \eqref{eq2-PDA-C2E}, we have $x=x_i$ for some
$x_i\in A_y$. Then by \eqref{eq1-PDA-C2E} and \eqref{eq3-PDA-C2E},
we have $\textbf{C}_{X, Z}(x_i,z_i)=1$ and $\textbf{C}_{Y,
Z}(y,z_i)=1$, where $z_i\in Z$ satisfying $(x_i,z_i)\in\Omega_y$.
Moreover, by condition C3, we have $\textbf{C}_{X,
Z}(x_i,z_{i'})=0$ for any $z_i\neq z_{i'}\in B_y$. So $z_i$ is the
unique point in $Z$ that satisfies $\textbf{C}_{X,
Z}(x_i,z)=\textbf{C}_{Y, Z}(y,z)=1$, which implies that
$\textbf{C}_{X, Y}$, $\textbf{C}_{X, Z}$ and $\textbf{C}_{Y, Z}$
satisfy condition E3. By the similar discussions, we can prove
that $\textbf{C}_{X, Y}$, $\textbf{C}_{X, Z}$ and $\textbf{C}_{Y,
Z}$ satisfy condition E5.

For any given $(x,z)\in X\times Z$ such that $\textbf{C}_{X,
Z}(x,z)=1$, by \eqref{eq1-PDA-C2E}, we have $(x,z)\in\Omega_y$ for
some $y\in Y$, so by \eqref{eq2-PDA-C2E} and \eqref{eq3-PDA-C2E},
we have $\textbf{C}_{X, Y}(x,y)=1$ and $\textbf{C}_{Y, Z}(y,z)=1$.
Now, suppose $y\neq y'\in Y$ such that $\textbf{C}_{X,
Y}(x,y')=\textbf{C}_{Y, Z}(y',z)=1$. Then by \eqref{eq2-PDA-C2E}
and \eqref{eq3-PDA-C2E}, we have $x\in A_{y'}$ and $z\in B_{y'}$.
Moreover, noticing that $\textbf{C}_{X, Z}(x,z)=1$, by condition
C3, we have $(x,z)\in\Omega_{y'}$, which contradicts to the fact
that $\Omega_y\cap\Omega_{y'}=\emptyset$ for any distinct $y',
y\in Y$. Hence, we proved that there is a unique $y\in Y$ that
satisfies $\textbf{C}_{X, Y}(x,y)=\textbf{C}_{Y, Z}(y,z)=1$. So
$\textbf{C}_{X, Y}$, $\textbf{C}_{X, Z}$ and $\textbf{C}_{Y, Z}$
satisfy condition E4.

By the above discussions, $\textbf{C}_{X, Y}$, $\textbf{C}_{X, Z}$
and $\textbf{C}_{Y, Z}$ satisfy conditions E1$-$E5 of Theorem
\ref{thm-PDA-SNC}, which completes the proof.
\end{proof}

\begin{proof}[Proof of Sufficiency]
Since $ Y$ is an $S$-set, without loss of generality, we can
assume that $Y=[S]$.

Suppose $\textbf{C}_{X, Y}$, $\textbf{C}_{X, Z}$ and
$\textbf{C}_{Y, Z}$ are binary matrices satisfying conditions
E1$-$E5 of Theorem \ref{thm-PDA-SNC}. Let $\textbf{P}_{X,
Z}=[p_{x,z}]_{X, Z}$ be an $F\times K$ array satisfying:
$p_{x,z}=*$ if $\textbf{C}_{ X, Z}(x,z)=0$, and $p_{x,z}=y$ if
$\textbf{C}_{ X, Z}(x,z)=\textbf{C}_{X, Y}(x,y)=\textbf{C}_{Y,
Z}(y,z)=1$. The construction is reasonable because $\textbf{C}_{X,
Y}$, $\textbf{C}_{X, Z}$ and $\textbf{C}_{Y, Z}$ satisfy condition
E4. We need to prove that $\textbf{P}_{X, Z}$ is a $(K,F,Q,S)$
PDA.

From the assumption that $\textbf{C}_{X, Y}$, $\textbf{C}_{X, Z}$
and $\textbf{C}_{Y, Z}$ satisfy conditions E1 and E2, we can
easily see that $\textbf{P}_{X, Z}$ satisfies conditions C1 and C2
of Definition \ref{defn-pda}. It remains to prove that
$\textbf{P}_{X, Z}$ satisfies condition C3.

Suppose $(x,z),(x',z')\in X\times Z$ such that $(x,z)\neq(x',z')$
and $p_{x,z}=p_{x',z'}=y$. Then by the construction of
$\textbf{P}_{X, Z}$, we have
\begin{align}\label{eq1-thm-PDA-SNC}\textbf{C}_{X, Y}(x,y)
=\textbf{C}_{X, Z}(x,z)=\textbf{C}_{Y, Z}(y,z)=1\end{align} and
\begin{align}\label{eq2-thm-PDA-SNC}
\textbf{C}_{X, Y}(x',y)=\textbf{C}_{X, Z}(x',z')=\textbf{C}_{Y,
Z}(y,z')=1.\end{align} The following discussions prove that
$\textbf{P}_{X, Z}$ satisfies condition C3 of Definition
\ref{defn-pda}.
\begin{itemize}
 \item Suppose $x'=x$. Since $(x,z)\neq(x',z')$, then
 $z'\neq z$. Moreover, by \eqref{eq1-thm-PDA-SNC} and
 \eqref{eq2-thm-PDA-SNC}, we have $$\textbf{C}_{X, Y}(x,y)=
 \textbf{C}_{X, Z}(x,z)=\textbf{C}_{Y, Z}(y,z)=1$$ and
 $$\textbf{C}_{X, Y}(x,y)=\textbf{C}_{X, Z}(x,z')
 =\textbf{C}_{Y, Z}(y,z')=1,$$ which
 contradicts to condition E3. So it must be the case that $x\neq
 x'$. Similarly, we can prove that $z\neq z'$ by contradiction to E5.
 \item Suppose $p_{x,z'}=y'$ for some $y'\in Y$. Then by
 construction of $\textbf{P}_{X, Z}$, we have
 \begin{align}\label{eq3-thm-PDA-SNC} \textbf{C}_{X,
 Y}(x,y')=\textbf{C}_{X, Z}(x,z')=\textbf{C}_{Y,
 Z}(y',z')=1.\end{align} Combining
 \eqref{eq1-thm-PDA-SNC}$-$\eqref{eq3-thm-PDA-SNC}, we have
 $$\textbf{C}_{X,
 Y}(x,y)=\textbf{C}_{X, Z}(x,z)=\textbf{C}_{Y, Z}(y,z)=1$$ and
 $$\textbf{C}_{X,
 Y}(x,y)=\textbf{C}_{X, Z}(x,z')=\textbf{C}_{Y, Z}(y,z')=1,$$ which
 contradicts to condition E3. So it must be the case that
 $p_{x,z'}=*$. Similarly, we can prove that $p_{x',z}=*$ by
 contradiction to E5.
\end{itemize}

So we proved that $\textbf{P}_{X, Z}$ satisfies condition C1$-$C3
of Definition \ref{defn-pda}, and hence is a $(K,F,Q,S)$ PDA,
which completes the proof.
\end{proof}

Another way to prove Theorem \ref{thm-PDA-SNC} is to prove firstly
that there exist three binary matrices $\textbf{C}_{ X, Y}$,
$\textbf{C}_{ X, Z}$, $\textbf{C}_{ Y, Z}$ satisfy conditions
E1$-$E5 if and only if there exists a linear, $(6,3)$-free,
$3$-uniform, $3$-partite hypergraph with three vertex parts $F$,
$K$, $S$ such that $|F| = F$, $|K| = K$, $|S| = S$, and each
vertex $k\in K$ is incident with exactly $F-Q$ edges. Then Theorem
\ref{thm-PDA-SNC} can be obtained from \cite[Theorem 10]{Chong18}.

\end{document}